\documentclass{article}
\usepackage{arxiv}
\usepackage{amsmath}
\usepackage{amssymb}
\usepackage{amsfonts}
\usepackage{algorithmic}
\usepackage{textcomp}
\usepackage{xcolor}
\usepackage{graphicx}
\usepackage{multicol}
\usepackage{times}
\usepackage{bbm}
\usepackage{dsfont}
\usepackage{color}
\usepackage[scr]{rsfso}
\usepackage{mathtools}
\usepackage{soul}

\graphicspath{{./figures/}}

\ifx\ieeestyle\undefined
\usepackage[utf8]{inputenc} 
\usepackage[T1]{fontenc}    
\usepackage{amsthm}
\usepackage[numbers,sort&compress]{natbib}
\usepackage[justification=centering]{subfig}
\usepackage{caption}
\usepackage{hyperref}
\usepackage[capitalise,nameinlink]{cleveref}

\theoremstyle{plain}
\newtheorem{theorem}{Theorem}[section]
\newtheorem{prop}{Proposition}[section]
\newtheorem{lemma}{Lemma}[section]
\newtheorem{assumption}{Assumption}[section]
\newtheorem{problem}{Problem}[section]
\newtheorem{remark}{Remark}[section]

\theoremstyle{definition}
\newtheorem{definition}{Definition}[section]
\newtheorem{example}{Example}[section]
\newtheorem{condition}{Condition}
\else
\usepackage[caption=false]{subfig}
\usepackage[capitalise,nameinlink]{cleveref}
\newtheorem{theorem}{Theorem}[section]

\newtheorem{lemma}{Lemma}[section]
\newtheorem{assumption}{Assumption}[section]

\newtheorem{remark}{Remark}[section]
\newtheorem{definition}{Definition}[section]

\fi

\crefname{assumption}{Assumption}{Assumptions}

\makeatletter
\newcommand{\printfnsymbol}[1]{%
  \textsuperscript{\@fnsymbol{#1}}%
}
\makeatother

\newcommand*\diff{\mathop{}\!\mathrm{d}}

\DeclareMathOperator*{\argmin}{arg\,min}
\DeclareMathOperator*{\ProjOp}{Proj}
\newcommand{\Proj}{\ProjOp\nolimits}
\DeclarePairedDelimiter\abs{\lvert}{\rvert}%
\DeclarePairedDelimiter\norm{\lVert}{\rVert}%

\makeatletter
\let\oldabs\abs
\def\abs{\@ifstar{\oldabs}{\oldabs*}}
\let\oldnorm\norm
\def\norm{\@ifstar{\oldnorm}{\oldnorm*}}
\makeatother

\newcommand{\ogamma}{\overline{\gamma}}
\newcommand{\oalpha}{\overline{\alpha}}
\newcommand{\ualpha}{\underline{\alpha}}
\newcommand{\olambda}{\overline{\lambda}}
\newcommand{\ulambda}{\underline{\lambda}}

\newcommand{\usigma}{\underline{\sigma}}
\newcommand\Sym[1]{\left[#1\right]_{\mathbb{S}}}

\newcommand{\pdv}[2]{\frac{\partial #1}{\partial #2}}

\newcommand{\odv}[2]{\frac{\diff #1}{\diff #2}}

\usepackage{tikz}
\usetikzlibrary{shapes,arrows.meta}
\usetikzlibrary{positioning}

\begin{document}

\title{Safe Feedback Motion Planning: A Contraction Theory and $\mathcal{L}_1$-Adaptive Control Based Approach}

\author{ Arun Lakshmanan\thanks{These authors contributed equally for this work} \\
    Mechanical Science and Engineering\\
	University of Illinois at Urbana-Champaign\\
	Urbana, IL 61801 \\
	\texttt{lakshma2@illinois.edu} \\
	\And
	Aditya Gahlawat\footnotemark[\value{footnote}] \\
    Mechanical Science and Engineering\\
	University of Illinois at Urbana-Champaign\\
	Urbana, IL 61801 \\
	\texttt{gahlawat@illinois.edu} \\
	\And
	Naira Hovakimyan \\
    Mechanical Science and Engineering\\
	University of Illinois at Urbana-Champaign\\
	Urbana, IL 61801 \\
	\texttt{nhovakim@illinois.edu} \\
}

\date{}

\renewcommand{\headeright}{}
\renewcommand{\undertitle}{}

\maketitle              

\begin{abstract}
Autonomous robots that are capable of operating safely in the presence of imperfect model knowledge or external disturbances are vital in safety-critical applications. In this paper, we present a planner-agnostic framework to design and certify safe tubes around desired trajectories that the robot is always guaranteed to remain inside. By leveraging recent results in contraction analysis and $\mathcal{L}_1$-adaptive control  we synthesize an architecture that induces safe tubes for nonlinear systems with state and time-varying uncertainties.  We demonstrate with a few illustrative\footnote{The implementation can be found \href{https://github.com/arlk/SafeFeedbackMotionPlanning.jl}{\texttt{https://github.com/arlk/SafeFeedbackMotionPlanning.jl}}} examples how contraction theory-based $\mathcal{L}_1$-adaptive control can be used in conjunction with traditional motion planning algorithms to obtain provably safe trajectories.
\end{abstract}

\section{Introduction}
\label{sec:intro}
Motion planning algorithms generate optimal  open-loop trajectories for robots to follow; however, any uncertainty in the system can potentially drive the robot far away from the desired path. For instance, quadrotors experience   blade-flapping and induced drag forces that are dependent on the velocity, ground effects that are dependent on the altitude, and external wind effects that are often unaccounted for by the motion planner, \cite{mahony2012multirotor}. Accurate modeling of these uncertainty effects on system dynamics can be very expensive and time-consuming.
A widely accepted approach to account for uncertainty in motion planning is through feedback \cite[Chapter~8]{lavalle2006planning}.  In practice, ancillary tracking controllers or model predictive control (MPC) schemes are employed to alleviate this problem. However, the presence of the uncertainties is not explicitly considered in the control design process, and instead the performance is achieved with hand-tuned controller parameters and experimental validation. Without valid safety certificates, the uncertainty might drive the system unstable and far enough away from the desired trajectory, resulting in collisions with obstacles,  \cref{fig:intro_none}.

Robust trajectory tracking controllers using classical Lyapunov stability theory have been designed for helicopters \cite{mahony2004robust}, hovercraft \cite{jeong2017coupled}, marine vehicles \cite{donaire2017trajectory}, and several other autonomous robots, which exhibit nonlinear behavior. These approaches rely on backstepping techniques, sliding-mode control, passivity-based control, or other robust nonlinear control design tools \cite[Chapter 14]{khalil2002nonlinear}. However, the classical methods do not provide a `one size fits all' procedure for the constructive design of tracking controllers for a large class of nonlinear systems. Unless the problem has a very specific structure that can be exploited, a control Lyapunov function (CLF) has to be found which can be prohibitively difficult for general nonlinear systems because the feasibility conditions do not appear as linear matrix inequalities (LMI), unlike in case of linear systems.

Advances in computational resources and optimization toolboxes available to autonomous robots have led to  active developments in the field of robust MPC. The two large classes of methods of interest are min-max MPC \cite{magni2001receding,raimondo2009min,wang2017adaptive} and tube-based MPC \cite{rakovic2009set,rakovic2016elastic,yu2010robust,williams2018robust,lopez2019dynamic}. Min-max MPC  approaches  consider the worst-case disturbance that can affect the system making them overly conservative. If the uncertainty is too large or the robot is planning over a long horizon, a min-max MPC based approach may even render the optimization infeasible. Tube-based MPC methods address these issues by employing an ancillary controller to attenuate disturbances and ensure that the robot stays inside of a `tube' around the desired trajectory. However, with the exception of \cite{lopez2019dynamic}, these methods assume the existence of a stabilizing ancillary controller and its region of attraction along the desired trajectory. Moreover, the resulting tubes are of fixed width, which may be overly conservative depending on the operating conditions (see \cref{fig:intro_wide}). This issue is partly addressed for feedback linearizable systems in \cite{lopez2019dynamic} by using sliding-mode boundary layer control to construct tubes of any desired size during the MPC optimization procedure. Furthermore, unlike classical methods, the MPC-based approaches while applicable to larger class of systems incur a heavy computational load and are not always amenable to real-time applications.

\begin{figure}[t]
    \centering
    \subfloat[]{\label{fig:intro_none}\includegraphics[width=0.33\columnwidth]{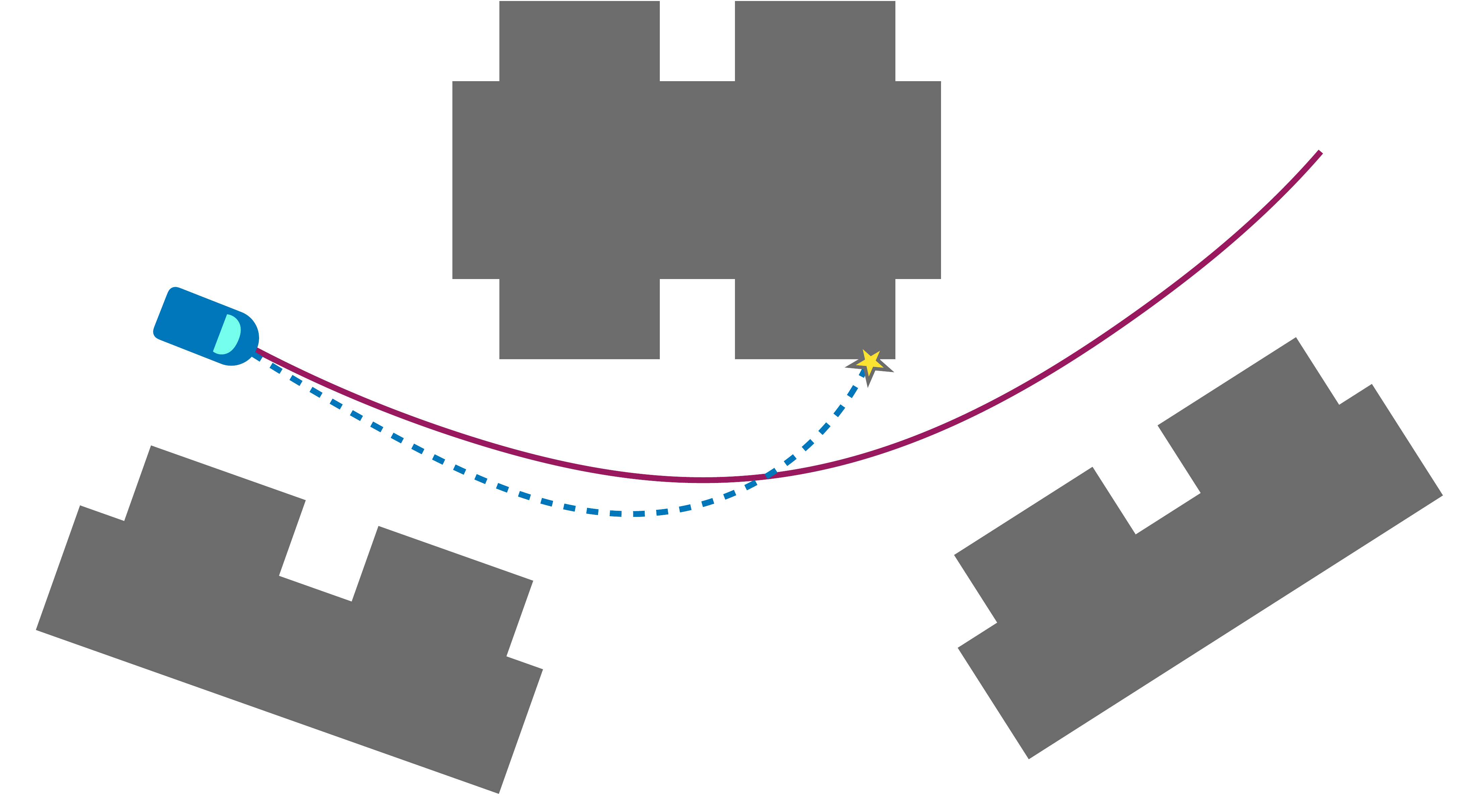}}
    \subfloat[]{\label{fig:intro_wide}\includegraphics[width=0.33\columnwidth]{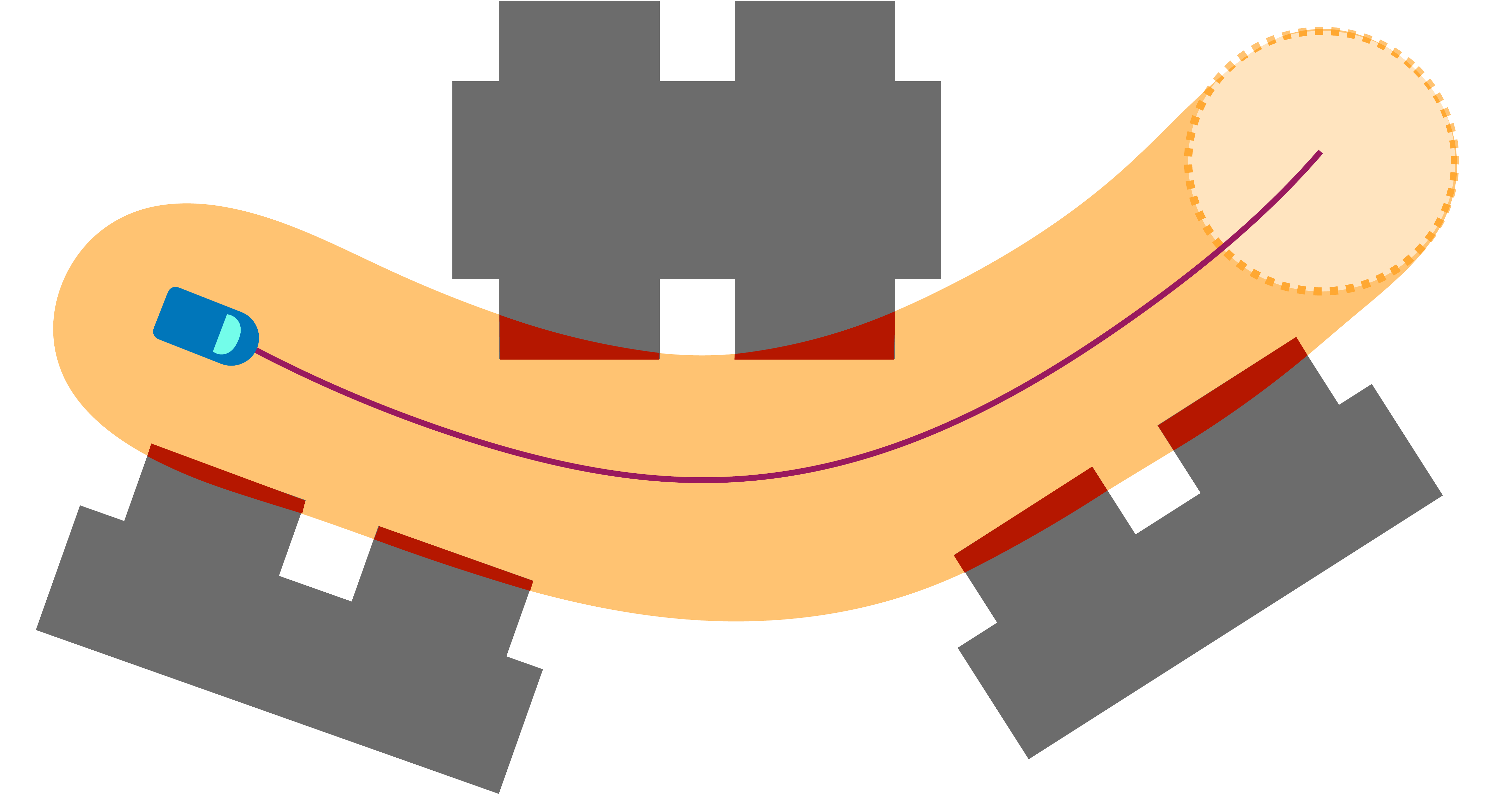}}
    \subfloat[]{\label{fig:intro_small}\includegraphics[width=0.33\columnwidth]{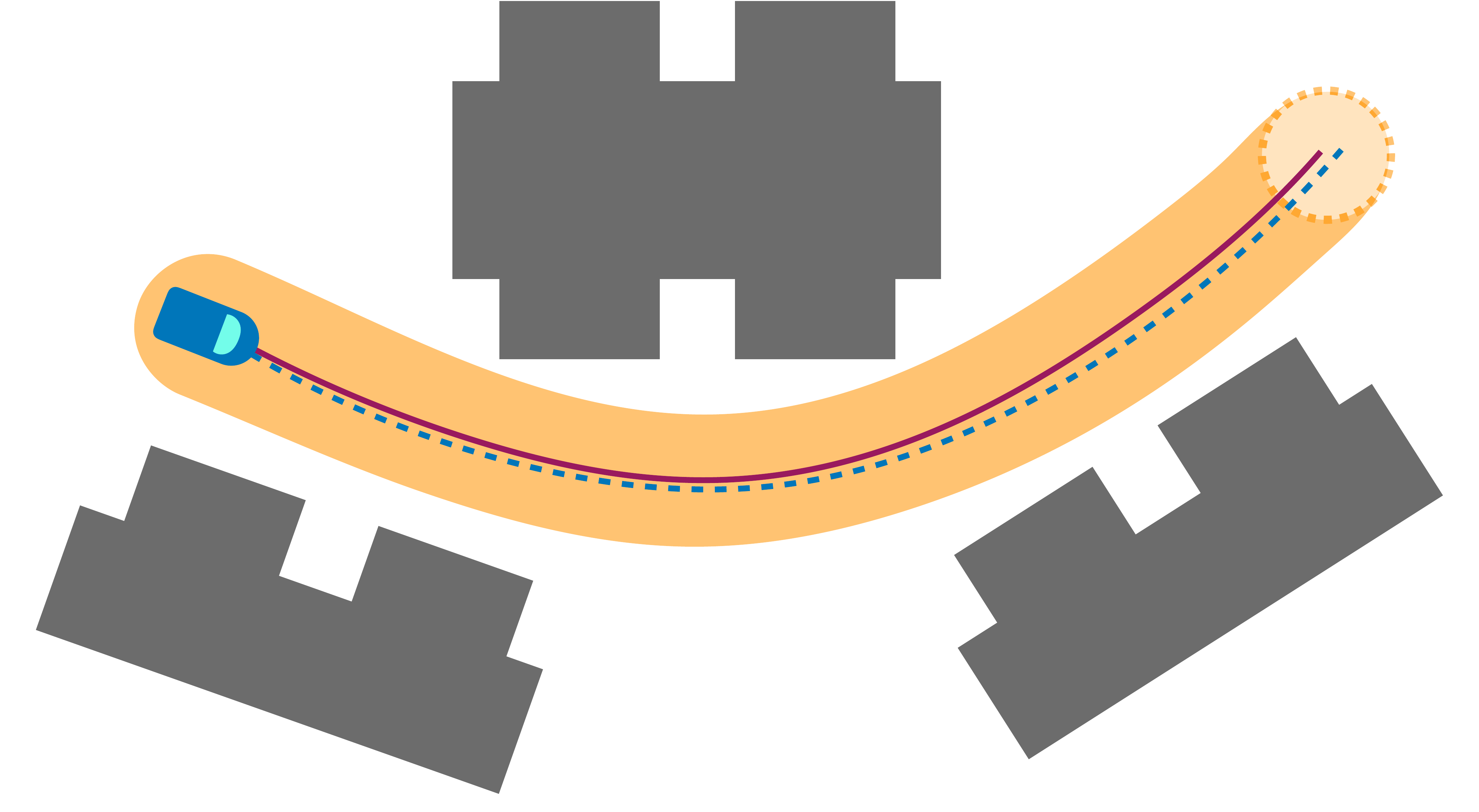}}
    \caption{Although the planned path is collision free (purple), the robot's actual trajectory (dashed-blue) might lead to a collision with the obstacles (gray) in the environment due to model discrepancies or external disturbances. (b) A feedback policy ensures that the robot stays inside of the (orange) tube which is too wide to pass between the obstacles without colliding (c) The safe feedback controller proposed in this paper guarantees that the robot's trajectory  never escapes the tube, which itself is also collision-free.}
    \label{fig:intro}
\end{figure}

Contraction theory-based approaches \cite{manchester2017control,manchester2018unifying} bridge the gap between classical and optimization-based methods, and provide a constructive control design procedure for nonlinear systems. In \cite{lohmiller1998contraction}, the authors introduce contraction analysis as tool for studying stability  of nonlinear systems using differential geometry. In particular, the authors show that the `contracting' or convergent nature of solutions to nonlinear systems can be derived from the differential dynamics of the system. Since the differential dynamics for nonlinear systems are of  linear-time varying (LTV) form, all the results from linear systems theory can be leveraged for nonlinear systems through the contraction analysis framework. In \cite{manchester2017control}, constructive control design techniques from linear systems theory can be used to find a control contraction metric (CCM), which is analogous to CLFs in the differential framework. This is significantly easier than directly finding the CLFs for nonlinear systems, because the feasibility conditions for CCMs are represented as LMIs. In \cite{singh2019robust},  a design procedure for synthesizing CCM-based controllers is given, which induces fixed-width tubes in the presence of bounded external disturbances, excluding modeling uncertainties. However, as discussed before, fixed-width tubes might result in infeasibility of the problem and result in  more work for the planner to find a more conservative path that produces feasible tubes. More recently, in \cite{lopez2019contraction} a model reference control architecture in conjunction with CCM-based feedback is proposed for handling uncertainties in the system.

In this paper, we present an approach for safe feedback motion planning for control-affine nonlinear systems that relies on contraction theory-based solution for exponential stabilizability around trajectories and $\mathcal{L}_1$-adaptive control for handling uncertainties and providing guarantees for transient performance and robustness. In $\mathcal{L}_1$ control architecture, estimation is decoupled from control, thereby allowing for arbitrarily fast adaptation subject  only to hardware limitations, \cite{hovakimyan2010L1}. The $\mathcal{L}_1$ control has been successfully implemented on NASA's AirStar 5.5\% subscale generic transport aircraft model \cite{gregory2010flight}, Calspan's Learjet \cite{ackerman2017evaluation}, and unmmaned aerial vehicles  \cite{kaminer2010path,jafarnejadsani2017optimized}.
In \cite{wang2012l1}, the authors presented the analysis for the $\mathcal{L}_1$-adaptive control architecture  with nonlinear time-varying reference systems. However, the stabilizability of the nominal nonlinear model and  associated safety certificates were simply assumed. In this paper, we present a constructive design of feedback strategy for nonlinear systems using CCMs and $\mathcal{L}_1$-adaptive control that provides strong guarantees of transient performance and robustness for a large class of control-affine nonlinear systems. Furthermore, we show how this control architecture induces tubes that can be flexibly changed to ensure safety based on the uncertainty in the system and the environment. In particular, this flexibility is provided by the architecture of the $\mathcal{L}_1$-adaptive control by decoupling the control loop from the estimation loop~\cite{hovakimyan2010L1}. In this way, the width of the certifiable tubes can be adjusted allowing the safe operation of a robot in tight confines.

The manuscript is organized as follows. The problem statement and the assumptions are provided in Section~\ref{sec:problem_statement}. A brief introduction to contraction theory is provided in Section~\ref{sec:prelim}. The proposed controller in presented in Section~\ref{sec:Contractive_L1} and the stability analysis of the closed-loop system is provided in Section~\ref{sec:performance_analysis}. Finally, in Section~\ref{sec:sim} the results of numerical experimentation are provided.

\subsection*{Notation}
The notation in this paper follows typical conventions used in the controls and robotics communities, but we list a few operations and terms explicitly in this section. The set of non-negative reals is denoted by $\mathbb{R}_{\geq 0}$. The set of real matrices  is denoted by $\mathbb{R}^{m \times n}$ with dimensions $m,n \in \mathbb{N}$. The set of real symmetric matrices is denoted by $\mathbb{S}^n \subset \mathbb{R}^{n \times n}$.
We denote an $n \times n$ identity matrix by $\mathbb{I}_n$ and a $m$-column vector of ones by $\mathds{1}_m$. Given any matrix $R \in \mathbb{R}^{n \times n}$, $\Sym{R} = \left(R + R^\top\right)/2 \in \mathbb{S}^n$ denotes its symmetric part.
The largest and smallest eigenvalue for a square matrix $A$ is denoted by $\overline{\lambda}(A)$ and $\underline{\lambda}(A)$ respectively. The positive definiteness of a square matrix $A$ is given by the inequality $A \succ 0$. The smallest non-zero singular value for a matrix $B \in \mathbb{R}^{n \times m}$ is denoted by $\usigma_{> 0}(B)$. We denote the Pontryagin set difference between two sets $\mathcal{A}, \mathcal{B} \subseteq \mathbb{R}^n$ as $\mathcal{A} \ominus \mathcal{B}$. The $\abs{\cdot}$ represents the absolute value of a scalar and $\norm{\cdot}$ represents the Euclidean norm for vectors and induced vector norm for matrices, unless otherwise denoted. The Laplace transform and its inverse of a function $f(t)$ is denoted by $\mathscr{L}[f(t)]$ and $\mathscr{L}^{-1}[f(t)]$ respectively. Given a matrix-valued function $M(x) \in \mathbb{R}^{n \times n}$ and a vector-valued function $f(x) \in \mathbb{R}^n$, we denote the directional derivative of $M(x)$ with respect to $f(x)$ by $\partial_f M(x) = \sum_{i = 1}^n (\partial M(x) / \partial {x_i}) f_i(x)$, where $f_i(x)$ is the $i^{th}$ element of $f(x)$. We denote by $\mathcal{L}_\infty(\mathcal{S})$ the set of functions $f:\mathcal{S} \subset \mathbb{R} \rightarrow \mathbb{R}^n$ which satisfy
$
\norm{f}_{\mathcal{L}_\infty} = \sup_{x \in \mathcal{S}} \norm{f(x)} < \infty,
$
where $\norm{\cdot}$ is a Euclidean norm. Additionally, for any $g \in \mathbb{R} \rightarrow \mathbb{R}^n$ we define the truncated $\mathcal{L}_\infty$-norm as
$
\norm{g}_{\mathcal{L}_\infty}^{[0, \tau]}  = \sup_{t \in [0, \tau]} \norm{g(t)},
$ for some finite $\tau > 0$.
We denote by $\mathcal{L}_1(\mathcal{S})$ the set of functions $f:\mathcal{S} \subset \mathbb{R} \rightarrow \mathbb{R}^n$ satisfying
$
\norm{f}_{\mathcal{L}_1} = \int_{\mathcal{S}} \norm{f(x)}dx < \infty.
$ For any $g \in \mathcal{L}_1(\mathbb{R}_{\ge 0})$, the truncated $\mathcal{L}_1$ norm is defined analogously.


\section{Problem Statement}\label{sec:problem_statement}

We consider systems for which the evolution of dynamics can be represented as
\begin{subequations}\label{eqn:proset:dynamics}
\begin{align}
    \dot{x}(t) = & F(x(t),u(t))  \\
    = & f(x(t)) + B(x(t))(u(t) + h(t,x(t))),
\end{align}
\end{subequations}
with initial condition $x(0) = x_0$, where $x(t) \in \mathbb{R}^n$ is the system state and $u(t) \in \mathbb{R}^m$ is the control input. The functions $f(x) \in \mathbb{R}^n$ and $B(x) \in \mathbb{R}^{n \times m}$ are known, and $h(t,x) \in \mathbb{R}^m$ represents the uncertainties.
The \textit{unperturbed/nominal dynamics} ($h \equiv 0$) are therefore represented as
\begin{subequations}\label{eqn:proset:unperturbed_dynamics}
\begin{align}
    \dot{x}(t) = & \bar{F}(x(t),u(t)) \\
    = & f(x(t)) + B(x(t))u(t), \quad x(0) = x_0.
\end{align}
\end{subequations}
Consider a \textit{desired control trajectory} $u^\star(t) \in \mathbb{R}^m$ and the induced \textit{desired state trajectory} $x^\star(t) \in \mathbb{R}^n$ from any planner based on unperturbed/nominal dynamics
\begin{equation}\label{eqn:proset:desired_trajectory}
    \dot{x}^\star(t) = \bar{F}(x^\star(t),u^\star(t)), \quad x^\star(0) = x^\star_0.
\end{equation} Together, $(x^\star(t),u^\star(t))$ is referred to as the \textit{desired state-input trajectory pair}. The planner ensures that the desired state-trajectory $x^\star(t)$ remains in a compact \textit{safe set} $\mathcal{X} \subset \mathbb{R}^n$, for all $t \geq 0$.

The goal is to design a control input $u(t)$ so that the state $x(t)$ of the uncertain system in~\eqref{eqn:proset:dynamics} remains `close' to the desired trajectory $x^\star(t)$ while also ensuring $x(t) \in \mathcal{X}$, for all $t \geq 0$. In order to rigorously define the notion of `closeness', we need the following definition:

\begin{definition}\label{def:robust_adaptive_tube}
Given a positive scalar $\rho$ and the desired state trajectory $x^\star(t)$, $\Omega(\rho,x^\star(t))$ denotes the $\rho$-norm ball around $x^\star(t)$, i.e.
\begin{equation}\label{eqn:omega_norm_ball}
\Omega(\rho,x^\star(t)) := \{y \in \mathbb{R}^n~|~\norm{y - x^\star(t)} \leq \rho\}.
\end{equation}
Clearly $\Omega(\rho,x^\star(t))$ induces a \textit{tube} centered around $x^\star(t)$, where the tube is given by
\begin{equation}
    \label{eqn:proset:full_tube}
    \mathcal{O}(\rho) := \bigcup_{t \geq 0} \Omega(\rho,x^\star(t)),
\end{equation}
with $\rho > 0$ as the radius.
\end{definition}

The problem under consideration can now be stated as follows: Given the desired trajectory $x^\star(t) \in \mathcal{X}$ and a positive scalar $\rho$, design a control input $u(t)$ such that the state of the uncertain system~\eqref{eqn:proset:dynamics} satisfies:
\[
x(t) \in \Omega(\rho,x^\star(t)) \subset \mathcal{X}, \quad \forall t \geq 0.
\] Note the condition that $\Omega(\rho,x^\star(t)) \subset \mathcal{X}$ is dependent on the desired trajectory $x^\star(t)$ (given by the planner) and the tube width $\rho$ (chosen by the user).
To ensure that this control-independent condition is satisfied, we place the following assumption.
\begin{assumption}\label{assmp:desired_traj_tube}
Given the positive scalar $\rho$, the desired state trajectory satisfies $x^\star(t) \in \mathcal{X}_\rho$, for all $t \geq 0$, where
\begin{equation}\label{eqn:potryagin_set}
\mathcal{X}_\rho = \mathcal{X} \ominus \mathcal{B}(\rho), \quad \mathcal{B}(\rho) := \{y \in \mathbb{R}^n~|~\|y\| \leq \rho\}. \end{equation}
\end{assumption}

\begin{remark}\label{rem:potryagin_difference}
The implication of \cref{assmp:desired_traj_tube} is that if the state trajectory satisfies $x(t) - x^\star(t) \in \mathcal{B}(\rho)$ and $x^\star(t) \in \mathcal{X}_\rho$, for all $t \geq 0$, then the definition of the Pontryagin set difference implies that $x(t) \in \Omega(\rho,x^\star(t)) \subset \mathcal{X}$, for all $t \geq 0$.
\end{remark}

\begin{assumption}\label{assmp:desired_control}
The desired control/input trajectory satisfies
\[
\norm{u^\star(t)} \leq \Delta_{u^\star}, \quad \forall t \geq 0,
\] with the upper bound $\Delta_{u^\star}$ known.
\end{assumption} Note that the bound $\Delta_{u^\star}$ is obtained from the planner, which provides the desired state-input trajectory in~\eqref{eqn:proset:desired_trajectory}.
Next, we place assumptions on the boundedness and continuity properties of the system functions and uncertainties.
\begin{assumption}\label{assmp:bounds_known}
The known functions  $f(x) \in \mathbb{R}^n$ and $B(x) \in \mathbb{R}^{n \times m}$ are bounded and continuously differentiable with bounded derivatives, satisfying
\begin{align*}
\|f(x)\| \leq \Delta_f, \quad \left\| \pdv{f(x)}{x}  \right\| \leq \Delta_{f_x}, \quad \|B(x)\| \leq \Delta_B,\quad
\sum_{i = 1}^n \left\|   \pdv{B(x)}{x_i}  \right\| \leq \Delta_{B_x}, \quad \sum_{j = 1}^m \left\|   \pdv{b_j(x)}{x}  \right\| \leq \Delta_{b_x}
\end{align*} for all $x \in \mathcal{O(\rho)}$, where $b_j(x)$ is the $j^\textrm{th}$ column of $B(x)$ and the bounds are assumed to be known.
\end{assumption}
\begin{assumption}\label{assmp:bounds_uncertainty}
The uncertainty $h(t,x)$ is bounded and continuously differentiable in both $x$ and $t$ with bounded derivatives, satisfying
\[
\|h(t,x)\| \leq \Delta_h, \quad  \norm{\pdv{h(t,x)}{x}} \leq \Delta_{h_x}, \quad  \norm{\pdv{h(t,x)}{t}} \leq \Delta_{h_t},
\] for all $x \in \mathcal{O(\rho)}$ and $t \geq 0$, where the bounds are assumed to be known.
\end{assumption}
\begin{assumption}\label{assmp:moore_penrose}
The input gain matrix $B(x)$ has full column rank. Furthermore, the Moore-Penrose inverse of $B(x)$ defined as $B^\dagger(x) = \left(B^\top (x) B(x)  \right)^{-1}B^\top(x)$ satisfies the following bounds
\[
\norm{B^\dagger(x)} \leq \Delta_{B^\dagger}, \quad \sum_{i=1}^n  \norm{\pdv{B^\dagger(x)}{x_i}} \leq \Delta_{B_x^\dagger}, \quad \forall x \in \mathcal{O(\rho)}.
\]
\end{assumption}


\section{Preliminaries on Contraction Theory}\label{sec:prelim}

Contraction theory allows  to synthesize feedback laws so that, in the absence of uncertainties, the state of the unperturbed/nominal dynamics in~\eqref{eqn:proset:unperturbed_dynamics} tracks a feasible desired trajectory $x^\star(t)$. We begin with the notion of universal exponential stabilizability.
\begin{definition}[\cite{singh2019robust}]\label{def:UES}
Consider a desired state-input trajectory pair $(x^\star(t),u^\star(t))$ satisfying~\cref{eqn:proset:desired_trajectory}. Suppose there exist scalars $\lambda,R > 0$ and a feedback operator $k_c: \mathbb{R}^n \times \mathbb{R}^n \rightarrow \mathbb{R}^m$ can be constructed such that the trajectory $x(t)$ of the unperturbed dynamics $\dot{x}(t) = \bar{F}(x(t),u_c(t))$ with control $u_c(t) = u^\star(t) + k_c(x^\star(t),x(t))$ satisfies
\[
\norm{x^\star(t) - x(t)} \leq R e^{-\lambda t} \norm{x^\star(0) - x(0)}, \quad \forall t \geq 0.
\] Then, the system with the unperturbed dynamics is said to be \textit{Universally Exponentially Stabilizable (UES)} with rate $\lambda$ and overshoot $R$.
\end{definition}

With the notion of UES defined, we now proceed to examine how UES may be established for a given system.
For the compact safe set $\mathcal{X} \subset \mathcal{R}^n$ defined in Section~\ref{sec:problem_statement}, let $T_x \mathcal{X}$ be the tangent space of $\mathcal{X}$ at $x \in \mathcal{X}$. Consequently, we denote by $T \mathcal{X} = \dot{\bigcup}_{x \in \mathcal{X}} T_x \mathcal{X}$ the tangent bundle of $\mathcal{X}$, where $\dot{\bigcup}$ denotes the disjoint union. Details on differential geometric notions used in the manuscript may be found in~\cite{bullo2019geometric}. The variational dynamics of  the unperturbed/nominal system in~\eqref{eqn:proset:unperturbed_dynamics}  may be written as~\cite[Chapter~3]{crouch1987variational}
\begin{align}\label{eqn:prelim:variational_dynamics}
    \dot{\delta}_x 
    = \left( \pdv{f(x)}{x} + \sum_{j = 1}^m u[j] \pdv{b_j(x)}{x} \right) \delta_x + B(x) \delta_u,
\end{align} with $\delta_x(0) = x_0$, where we have dropped the temporal dependencies for brevity. Here, $\delta_x(t) \in T_{x(t)} \mathcal{X}$, $\delta_u(t) \in T_{u(t)}\mathbb{R}^m$, $u[j](t)$ is the $j^{th}$ element of the control vector and $b_j(x) \in \mathbb{R}^n$ is the $j^{th}$ column of $B(x)$.

\begin{definition}\label{def:CCM}
Consider the differential dynamics in~\eqref{eqn:prelim:variational_dynamics}.  Suppose there exist positive scalars $\lambda$, $\underline{\alpha}$, $\overline{\alpha}$, $0 < \underline{\alpha} < \overline{\alpha} < \infty$, and a smooth\footnote{Throughout the manuscript, by smooth we mean the class $\mathcal{C}^\infty$ of functions defined on appropriate domains.} function $M: \mathbb{R}^n \rightarrow \mathbb{S}^n$ such that for all $(x,\delta_x) \in T \mathcal{X}$ one has
\begin{subequations}\label{eqn:prelim:CCM_conditions}
\begin{align}
   &\label{eqn:prelim:CCM_1}\underline{\alpha} \mathbb{I}_n \preceq M(x) \preceq \mathbb{I}_n \overline{\alpha}, \\
   &\delta_x^\top M(x)B(x) = 0 \Rightarrow  \notag \\
   &\label{eqn:prelim:CCM_2} \delta_x^\top \left(  \partial_f M(x) \hspace{-1mm} + \Sym{M(x) \pdv{f(x)}{x}} + 2 \lambda M(x) \right) \delta_x \leq 0,\\
   &\label{eqn:prelim:CCM_3} \partial_{b_j}M(x) + \Sym{M(x)\pdv{b_j(x)}{x}} = 0, \quad j \in \{1,\dots,m\}.
\end{align}
\end{subequations} Then, the function $M(x)$ is defined to be the \textit{Control Contraction Metric (CCM)} for the nominal/unperturbed dynamics~\eqref{eqn:proset:unperturbed_dynamics}.
\end{definition}

\begin{theorem}[\cite{singh2019robust,manchester2017control}]\label{thm:contraction_core}
Given positive scalars $\lambda$, and $\underline{\alpha} \leq \overline{\alpha} < \infty$, suppose there exists a CCM $M(x)$ for the nominal/unperturbed dynamics in~\cref{eqn:proset:unperturbed_dynamics}.
Then, given any desired state-input trajectory $(x^\star(t),u^\star(t))$ as in~\cref{eqn:proset:desired_trajectory}, there exists a feedback operator $k_c: \mathbb{R}^n \times \mathbb{R}^n \rightarrow \mathbb{R}^m$ such that the trajectory $x(t)$ of the unperturbed dynamics $\dot{x}(t) = \bar{F}(x(t),u_c(t))$ with control $u(t) = u^\star(t) + k_c(x^\star(t),x(t))$ is UES with respect to $x^\star(t)$ with the overshoot of $R = \oalpha/\ualpha$ in the sense of~\cref{def:UES}.
\end{theorem} The central idea to this result is that the function $V(x,\delta_x) := \delta_x^\top M(x) \delta_x$ can be interpreted as a differential Lyapunov function and the conditions in~\cref{eqn:prelim:CCM_conditions} ensure that $\dot{V}(x,\delta_x) \leq -  2\lambda V(x,\delta_x)$ for all $(x,\delta_x) \in T \mathcal{X}$.
We place the following assumption on the known/unperturbed dynamics.
\begin{assumption}\label{assmp:CCM}
The nominal/unperturbed dynamics in~\cref{eqn:proset:unperturbed_dynamics} admit a CCM $M(x)$ for all $x \in \mathcal{X}$ with positive scalars $\lambda$, $\underline{\alpha}$, and $\overline{\alpha}$, as in~\cref{def:CCM}.
\end{assumption} Using~\cref{thm:contraction_core} it is straightforward to conclude that the consequence of this assumption is that any desired state-input trajectory can be tracked by the nominal/unperturbed dynamics in the sense of~\cref{def:UES} with rate $\lambda$ and overshoot $R = \overline{\alpha}/\underline{\alpha}$. 
Let $\Xi(p,q)$ be the set of smooth curves connecting any two points $p,q \in \mathcal{X}$. Then using the Riemannian metric $M$, the length of any curve $\gamma \in \Xi(p,q)$ is given by the following expression
\begin{equation}
\label{eqn:prelim:length}
l(\gamma) := \int_0^1 \sqrt{\gamma_s(s)M(\gamma(s))\gamma_s(s)}\diff s,
\end{equation}
where $\gamma_s(s) = \partial \gamma(s) / \partial s$. By definition, the minimizing geodesic $\ogamma :[0,1] \rightarrow \mathcal{X}$ satisfies the following relationship
\begin{equation}
\label{eqn:prelim:distance}
d(p,q) := l(\ogamma) = \inf_{\gamma \in \Xi(p,q)}l(\gamma),
\end{equation}
where $d(p,q)$ refers to the Riemannian distance between the two points $p$ and $q$. Existence of the minimizing geodesic is guaranteed by the Hopf-Rinow theorem. The Riemannian energy between the two points is defined using the Riemannian distance as the following quantity
\begin{equation}
\label{eqn:prelim:energy}
\mathcal{E}(p,q) := d(p,q)^2.
\end{equation}
Further details on Riemannian geometry may be found in~\cite{carmo1992riemannian}. A direct and straightforward consequence of~\cref{assmp:CCM} is that
\begin{equation}\label{eqn:Riemannian_energy_bound}
    \underline{\alpha}\norm{p - q}^2 \leq \mathcal{E}(p,q) \leq \overline{\alpha} \norm{p - q}^2, \quad \forall p, q \in \mathcal{X}.
\end{equation} 
The proof for this relationship can be found in \cref{lem:riemannbounds}. We will rely on the Riemannian energy's interpretation as a control Lyapunov function for the presented methodology. This interpretation was initially presented in~\cite{manchester2017control}.

\begin{remark}
Thus far we have only established the existence of feedback control operators and not constructed any. In fact, as explained in~\cite[Sec.~VI.A]{manchester2017control}, any controller may be chosen as long as the following set membership is established
\[
\mathcal{U} = \{ u \in \mathbb{R}^m \ | \ \dot{\mathcal{E}}(x^\star(t), x(t)) \le -2\lambda \mathcal{E}(x^\star(t), x(t))\}.
\]
The precise choice of the controller we use will be presented later in the manuscript.
\end{remark}


\section{Contraction Theory Based $\mathcal{L}_1$-Adaptive Control}\label{sec:Contractive_L1}

In this section we introduce the structure of the proposed controller for the uncertain nonlinear system in~\cref{eqn:proset:dynamics}. Consider the following feedback decomposition
\begin{equation}\label{eqn:control_def:overall_control}
    u(t) = u_c(t) + u_a(t),
\end{equation} where $u_c: \mathbb{R}_{\geq 0} \rightarrow \mathbb{R}^m$ is the contraction theory based control designed to guarantee UES (\cref{def:UES}) of the nominal dynamics in~\cref{eqn:proset:unperturbed_dynamics}, and $u_a: \mathbb{R}_{\geq 0} \rightarrow \mathbb{R}^m$ is the $\mathcal{L}_1$ control signal. The overall architecture of the proposed feedback is illustrated in~\cref{fig:arch}. We refer to the uncertain system in~\cref{eqn:proset:dynamics} with the feedback law~\cref{eqn:control_def:overall_control} as the $\mathcal{L}_1$ closed-loop system. Before we proceed with the description of the individual components of the controller, we introduce the following list of constants that are of importance for the results and analyses presented in this paper:
\begin{align}
\label{eqn:bounds:M_x}
\Delta_{M_x} &:= \sup_{x \in \mathcal{O}(\rho)} \sum_{i=1}^n \norm{\pdv{M}{x_i}(x)}, \\
\label{eqn:bounds:psi_x}
\Delta_{\Psi_x} &:= 2\Delta_{B_x} + \frac{\Delta_B \Delta_{M_x}}{\ualpha},\\
\label{eqn:bounds:delta_u}
\Delta_{\delta_u} &:= \frac{1}{2}\sup_{x \in \mathcal{O}(\rho)}\left(\frac{\olambda(L^{-\top}(x)F(x)L^{-1}(x))}{\usigma_{> 0}(B^\top(x)L^{-1}(x))}\right),\\
 \label{eqn:bounds:dx_r}
\Delta_{\dot{x}_r} &:= \Delta_f + \Delta_B(\norm{\mathbb{I}_m - C(s)}_{\mathcal{L}_1}\Delta_h + \Delta_{u^\star} + \rho\Delta_{\delta_u}), \\
\label{eqn:bounds:dx}
\Delta_{\dot{x}} &:= \Delta_f + \Delta_B(2\Delta_h + \Delta_{u^\star} + \rho\Delta_{\delta_u}),\\
\label{eqn:bounds:tx}
\Delta_{\tilde{x}} &:= \sqrt{\frac{4\olambda(P)\Delta_h(\Delta_{h_t} + \Delta_{h_x}\Delta_{\dot{x}})}{\ulambda(P)\underline{\lambda}(Q)} + \frac{4\Delta_h^2}{\underline{\lambda}(P)}}, \\
\label{eqn:bounds:teta}
\Delta_{\tilde{\eta}} &:= \left(\Delta_{B^\dagger_x} \Delta_{\dot{x}} + (\norm{sC(s)}_{\mathcal{L}_1} + \norm{A_m}) \Delta_{B^\dagger} \right)\Delta_{\tilde{x}},\\
\label{eqn:bounds:theta}
\Delta_\theta &:= \frac{\Delta_B \overline{\alpha}  \Delta_{\tilde{\eta}}}{\lambda},\\
\label{eqn:bounds:dpsi}
\Delta_{\dot{\Psi}} &:= \oalpha\left(
    \Delta_B \Delta_{\dot{\ogamma}_s} + \frac{\Delta_B \Delta_{M_x} \Delta_{\dot{x}}}{\sqrt{\oalpha\ualpha}} + \Delta_{B_x} \Delta_{\dot{x}}\right), \\
\label{eqn:bounds:dgamma_s}
\Delta_{\dot{\ogamma}_s} &:= \sqrt{\frac{\oalpha}{\ualpha}}\left(\Delta_{f_x} + (\Delta_h + \Delta_{u^\star} + \rho \Delta_{\delta_u})\Delta_{b_x} + \left(\Delta_{h_x} + \frac{\sqrt{\ualpha}\Delta_{\delta_u}}{\sqrt{\oalpha}}\right) \Delta_B \right),
\end{align}
where $\mathcal{O}(\rho)$ is defined in \cref{eqn:proset:full_tube}; $\Delta_{u^\star}$ is defined in \cref{assmp:desired_control}; $\Delta_f$, $\Delta_{f_x}$, $\Delta_B$, $\Delta_{B_x}$, $\Delta_{b_x}$, are defined in \cref{assmp:bounds_known}; $\Delta_h$, $\Delta_{h_t}$, $\Delta_{h_x}$ are defined in \cref{assmp:bounds_uncertainty}; $\Delta_{B^\dagger}$ and $\Delta_{B_x^\dagger}$ are defined in \cref{assmp:moore_penrose}; $\oalpha$ and $\ualpha$ are defined in \cref{assmp:CCM}; and $F(x)$ is defined as
\[
F(x):=-\partial_{f} W(x)+ 2\Sym{\pdv{f}{x}(x)W(x)} + 2 \lambda W(x),
\]
where $W(x) = M(x)^{-1}$ is referred to as the dual metric and $L(x)^\top L(x) = W(x)$. 

\definecolor{lcolor}{RGB}{86,193,255}
\definecolor{bcolor}{RGB}{252,184,203}
\definecolor{ccolor}{RGB}{234,208,255}
\tikzstyle{sqblock} = [draw, fill=bcolor!20, rectangle,
    minimum height=3em, minimum width=3em]
\tikzstyle{block} = [draw, fill=bcolor!20, rectangle,
    minimum height=3em, minimum width=6em]
\tikzstyle{sum} = [draw, fill=bcolor!20, circle, node distance=1cm]
\tikzstyle{input} = [coordinate]
\tikzstyle{output} = [coordinate]
\tikzstyle{phantom} = [coordinate]

\begin{figure}
  \centering
    \begin{tikzpicture}[align=center]
    	\filldraw[fill=ccolor!20,draw=ccolor!20] (1.1, -2.1) rectangle (4.5,-1.1) node[pos=.5] {Contraction theory\\ based controller};
    	\filldraw[fill=lcolor!20,draw=lcolor!20] (4.9, -5.6) rectangle (12.2,-4.6) node[pos=.5] {$\mathcal{L}_1$-adaptive controller};
    	\filldraw[fill=ccolor!70, densely dotted] (1.1, -1.1) rectangle (4.5,1.2);
    	\filldraw[fill=lcolor!50, densely dotted] (4.9, -4.6) rectangle (12.2,-1.4);
        \node [input, name=input] {};
        \node [block] (ccm) [right=1.5cm of input] {CCM Feedback};
        \node [sum] (sum) [right=1.5cm of ccm] {};
        \node [block] (sys) [right=1.5cm of sum] {Uncertain System};
        \node [sqblock] (filt) [below=1.5cm of sum] {$C(s)$};
        \node [block] (pred) [right=1.5cm of filt] {State Predictor};
        \node [block] (adap) [below=0.5cm of pred] {Adaptation Law};
        \node [sum] (diff) [right=1.0cm of pred] {};
        \node [phantom, right=0.75cm of filt] (p1) {};
        \node [phantom] (p2) at (sys -| diff) {};
        \node [phantom, above=1.0cm of p2] (p3) {};

        \draw[-Latex] (input) -- node [near start,above] {$x^\star, u^\star$} (ccm);
        \draw[-Latex] (ccm) -- node [midway,above] {$u_c$} (sum);
        \draw[-Latex] (sum) -- node [midway,above] {$u$} (sys);
        \draw[-Latex] (filt) -- node [midway,left] {$u_a$} node[pos=0.99,left] {$+$} (sum);
        \draw (adap) -| node [near end,left] {$\hat{\sigma}$}(p1);
        \draw[-Latex] (p1) -- (pred);
        \draw[-Latex] (p1) -- (filt);
        \draw[-Latex] (pred) -- node [midway,above] {$\hat{x}$} (diff);
        \draw[-Latex] (diff) |- node [near start,right] {$\tilde{x}$} (adap);
        \draw[-Latex] (p2) -- node[pos=0.99,right] {$-$} (diff);
        \draw (p2) -- (p3);
        \draw (sys) -- node [midway,above] {$x$} (p2);
        \draw[-Latex] (p3) -| (ccm);
        \node [output, right=1.0cm of p2] (output) {};
        \draw[-Latex] (p2) -- (output);
    \end{tikzpicture}
  \caption{Architecture of CCM-based $\mathcal{L}_1$-adaptive control}
  \label{fig:arch}
\end{figure}
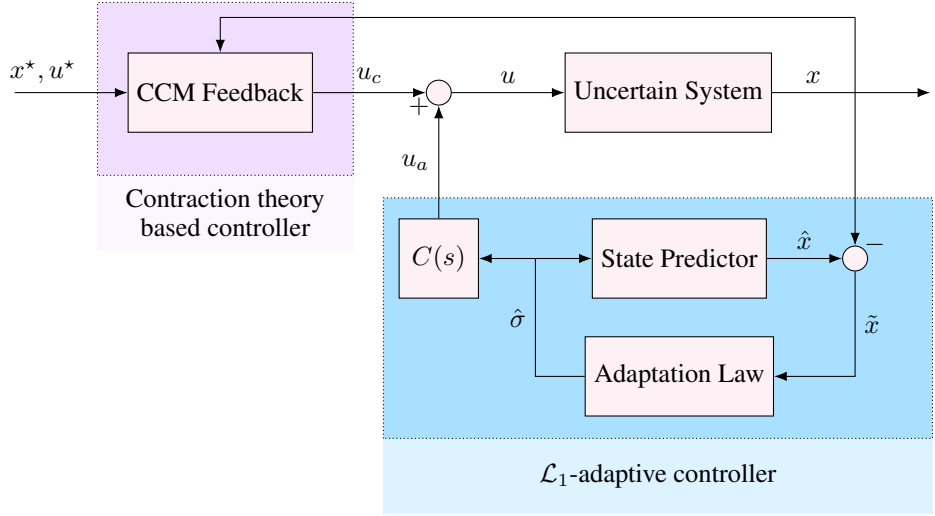

\subsection{Contraction theory based control:~$u_c(t)$}

As mentioned in~\cref{sec:prelim}, under \cref{assmp:CCM},~\cref{thm:contraction_core} guarantees the existence of a feedback law which renders the nominal dynamics in~\cref{eqn:proset:unperturbed_dynamics} UES. In particular, we propose the following law
\begin{equation}\label{eqn:control_def:u_c}
u_c(t) = u^\star(t) + k_c(x^\star(t),x(t)),
\end{equation} where, for the the feedback term, we use the law constructed in~\cite[Sec.~5.1]{singh2019robust}, which is the solution to the following quadratic program:
\begin{subequations}\label{eqn:control_def:u_c_QP}
\begin{align}
     & k_c(x^\star(t),x(t)) =  \argmin_{k \in \mathbb{R}^m} \ \norm{k}^2, \\
     &\text{s.t. }~2\ogamma_s^\top (1,t)M(x(t))\dot{x}_k(t) - 2 \ogamma_s^\top (0,t) M(x^\star(t)) \dot{x}^\star(t)
     \leq -2 \lambda \mathcal{E}(x^\star(t),x(t)),
\end{align}
\end{subequations} in which $M(\cdot)$ is the CCM (\cref{def:CCM}), $\ogamma(s,t)$, $s \in [0,1]$, is the minimizing geodesic with $\ogamma(1,t) = x_k(t)$ and $\ogamma(0,t) = x^\star(t)$. As previously defined, the desired state-input pair satisfies $\dot{x}^\star(t) = \bar{F}(x^\star(t),u^\star(t))$ with the nominal dynamics defined in~\cref{eqn:proset:unperturbed_dynamics}. Additionally, $\dot{x}_k(t) = \bar{F}(x(t),u^\star(t) + k)$.
\begin{remark}
As explained by the authors in~\cite[Sec.~5.1]{singh2019robust}, the solution to the quadratic program in~\eqref{eqn:control_def:u_c_QP} can be obtained analytically given the minimizing geodesic $\ogamma(\cdot,t)$. Alternatively, one may use the differential controller proposed in~\cite{manchester2017control}, albeit at the expense of an increase in the control effort.
\end{remark}
\subsection{$\mathcal{L}_1$-adaptive control:~$u_a(t)$}
\label{subsec:l1}

The computation of the signal $u_a(t)$ depends on three components illustrated in~\cref{fig:arch}, namely, the state-predictor, the adaptation law, and a low-pass filter. Similar to~\cite{wang2017adaptive}, we define the \textit{state-predictor} as
\begin{equation}\label{eqn:control_def:state_predictor}
 \dot{\hat{x}}(t) = \bar{F}(x(t),u_c(t) + u_a(t) + \hat{\sigma}(t)) + A_m \tilde{x}(t),
\end{equation} with $\hat{x}(0) = x_0$, and where $\hat{x}(t) \in \mathbb{R}^n$ is the state of the predictor, $\tilde{x}(t) = \hat{x}(t) - x(t)$ is the state prediction error, and $A_m \in \mathbb{R}^{n \times n}$ is an arbitrary Hurwitz matrix.

The \textit{uncertainty estimate} $\hat{\sigma}(t)$ in~\cref{eqn:control_def:state_predictor} is governed by the following \textit{adaptation law}
\begin{equation}\label{eqn:control_def:adaptation_law}
  \dot{\hat{\sigma}}(t) =  \Gamma \Proj_{\mathcal{H}}(\hat{\sigma}(t),-B(x)^\top P \tilde{x}(t)), \quad \hat{\sigma}(0) \in \mathcal{H},
\end{equation} where $\Gamma > 0$ is the adaptation rate, $\mathcal{H} = \{y \in \mathbb{R}^m~|~\norm{y} \leq \Delta_h\}$ is the set to which the uncertainty estimate is restricted to remain in with $\Delta_h$ defined in~\cref{assmp:bounds_uncertainty}. Furthermore, $\mathbb{S}^n \ni P \succ 0$, is the solution to the Lyapunov equation $A_m^\top P + P A_m = - Q$, for some $\mathbb{S}^n \ni Q \succ 0$. Moreover, $\Proj_{\mathcal{H}}(\cdot,\cdot)$ is the projection operator standard in adaptive control literature~\cite{lavretsky2011projection},~\cite{pomet1992adaptive}.

Finally, the control law $u_a(t)$ is defined as the following Laplace transform
\begin{equation}\label{eqn:control_def:control_law}
    u_a(s) = -C(s)\hat{\sigma}(s),
\end{equation} where $C(s)$ is a low-pass filter with bandwidth $\omega$ and satisfies $C(0) = \mathbb{I}_m$. Note that there is an abuse of notation when we denote both the geodesic interval parameter and the Laplace variable by $s$. The delineation between the two is clear from the context.

\subsection{Filter bandwidth and adaptation rate}\label{subsec:bandwidth_rate}

The design of the $\mathcal{L}_1$-adaptive controller involves the design of a strictly proper and stable low-pass filter $C(s)$ with $C(0) = \mathbb{I}_m$. Let the bandwidth of this filter be $\omega$. In the manuscript, for the sake of simplicity, we choose $C(s) = \frac{\omega}{s + \omega} \mathbb{I}_m$.
As we will see in~\cref{sec:performance_analysis}, the bandwidth $\omega$ of the low-pass filter $C(s)$ in~\cref{eqn:control_def:control_law} and the adaptation rate $\Gamma$ in~\cref{eqn:control_def:adaptation_law} are design parameters which can be thought of as `tuning-knobs'. However, these entities need to satisfy a few conditions mentioned below. The reasoning behind these conditions will be made clear in the subsequent section.

Suppose that~\cref{assmp:CCM} holds. Then, for arbitrarily chosen positive scalars $\epsilon$ and $\rho_a$, define
\begin{align}
\label{eqn:conditions:rho_r}
\rho_r &= \sqrt{\frac{\overline{\alpha}}{\underline{\alpha}}} \norm{x_0^\star - x_0} + \epsilon,\\
\label{eqn:conditions:rho}
\rho &= \rho_r + \rho_a.
\end{align}
Furthermore, suppose that Assumptions~\ref{assmp:desired_traj_tube}-\ref{assmp:moore_penrose} hold. Define
\begin{subequations}\label{eqn:conditions:zeta_definitions}
\begin{align}
    \zeta_1(\omega)
    =&\label{eqn:conditions:zeta_1} 2 \rho  \Delta_B\frac{\overline{\alpha}}{\ualpha} \left(\frac{\Delta_h}{\abs{2 \lambda - \omega}} + \frac{\Delta_{h_t} + \Delta_{h_x} \Delta_{\dot{x}_r}}{2 \lambda \omega}  \right),     \\
    \zeta_2(\omega) = &\label{eqn:conditions:zeta_2}  \overline{\alpha} \Delta_{\Psi_x}  \left(\frac{\Delta_h}{\abs{2 \lambda - \omega}} + \frac{\Delta_{h_t} + \Delta_{h_x} \Delta_{\dot{x}_r}}{2 \lambda \omega}  \right), \\
    \zeta_3(\omega) = &\label{eqn:conditions:zeta_3} \overline{\alpha} \Delta_{h_x} \left( \frac{4 \lambda \Delta_B + \Delta_{\dot{\Psi}}}{\lambda \omega} \right),
\end{align} where $\Delta_{\dot{x}_r}$, $\Delta_{\Psi_x}$, and $\Delta_{\dot{\Psi}}$,
are known positive scalars defined in \cref{eqn:bounds:dx_r,eqn:bounds:psi_x,eqn:bounds:dpsi} respectively.
\end{subequations} Then, the bandwidth $\omega$ of the low-pass filter $C(s)$ and the adaptation rate need to verify the following conditions
\begin{subequations}\label{eqn:bandwidth_rate_conditions}
\begin{align}
    \rho_r^2 \geq & \label{eqn:filter_condition_1} \frac{\mathcal{E}(x_0^\star,x_0)}{\underline{\alpha}} + \zeta_1(\omega),\\
    \underline{\alpha} > & \label{eqn:filter_condition_2} \zeta_2(\omega) + \zeta_3(\omega),\\
    \sqrt{\Gamma} >& \label{eqn:adaptation_rate_condition} \frac{\Delta_\theta}{ \rho_a (\underline{\alpha} - \zeta_2(\omega) - \zeta_3(\omega) )},
\end{align}
\end{subequations} where $\Delta_\theta$ is another known positive scalar defined in \cref{eqn:bounds:theta}.
\begin{remark}\label{rem:bandwidth_conditions}
Based on the definition of $\rho_r$ in~\cref{eqn:conditions:rho} and the bounds on the Riemannian energy $\mathcal{E}(x^\star(t),x(t))$ in~\cref{eqn:Riemannian_energy_bound}, the inequality $\rho_r^2 > \mathcal{E}(x_0^\star,x_0)/\underline{\alpha}$ holds. Furthermore, since $\zeta_1(\omega)$, $\zeta_2(\omega)$, and $\zeta_3(\omega)$, all converge to zero as $\omega$ increases, the bandwidth conditions in~\eqref{eqn:filter_condition_1}-\eqref{eqn:filter_condition_2} can always be satisfied by choosing a large enough $\omega$.
\end{remark}


\section{Performance Analysis}\label{sec:performance_analysis}

In this section we analyze the performance of the uncertain system in~\cref{eqn:proset:dynamics} with the $\mathcal{L}_1$ control feedback $u(t)$ defined in~\cref{eqn:control_def:overall_control}. As in~\cite{wang2017adaptive}, to derive the bounds between the desired trajectory $x^\star(t)$ and the state $x(t)$ of the uncertain system, we first introduce the following intermediate system, which we refer to as the \textit{reference system}:
\begin{subequations}\label{eqn:reference_system}
\begin{align}
    \dot{x}_r(t) =& F(x_r(t),-\eta_r(t)) = f(x_r(t)) + B(x_r(t))(u_{c,r}(t) -\eta_r(t) + h(t,x_r(t))), \\
    u_{c,r}(t) = & u^\star(t) + k_c(x^\star(t),x_r(t)),\\
    \eta_r(s) = & C(s) \mathscr{L}[h(t,x_r(t))], \quad x_r(0) = x_0,
\end{align}
\end{subequations} where $k_c$ is defined in \cref{eqn:control_def:u_c_QP} using $x_r$ in place of $x$.
The main feature of the reference system is that it defines the {\em the best achievable performance}, given the perfect knowledge of uncertainty, i.e. it reflects that the cancellation of the uncertainty $h(t,x_r(t))$ can happen only within the bandwidth of the low-pass filter.

The analysis consists of two parts: we first derive bounds between the desired trajectory and the reference system $\|x^\star(t) - x_r(t)\|$. Then we derive the bounds between the states of the reference system and the actual system $\norm{x_r(t) - x(t)}$. Recall that  we refer to the actual system as the $\mathcal{L}_1$ closed loop system, which is given by~\cref{eqn:proset:dynamics} with the control law in~\cref{eqn:control_def:overall_control}.
Finally, the triangle inequality  produces the desired bound on $\norm{x^\star(t) - x(t)}$. In this way, the reference system behaves as an `anchor system' for the analysis. These bounds are illustrated in~\cref{fig:analysis_bounds}. Furthermore, we provide the justification of treating the bandwidth $\omega$ of $C(s)$ and the adaptation rate $\Gamma$ as tuning-knobs. Indeed, the upcoming analysis will show that we can ensure that $x(t) \in \Omega(\rho,x^\star(t))$ (see~\cref{eqn:omega_norm_ball}) for all $t \geq 0$.

We begin with the bound between the reference system state and desired state trajectory. This corresponds to the green tube in~\cref{fig:analysis_bounds}. The proofs for all the claims in this section are provided in Appendix~\ref{app:main}.
\begin{lemma}\label{lem:refbound}
Let all the assumptions hold and let $\rho_r$ be as defined in~\cref{eqn:conditions:rho}. If the conditions in~\eqref{eqn:filter_condition_1}-\eqref{eqn:filter_condition_2} hold, then for any desired state trajectory $x^\star(t)$ the state $x_r(t)$ of the reference system in~\eqref{eqn:reference_system} satisfies
\begin{equation}\label{eqn:ref:uniform_tube}
x_r(t) \in \Omega(\rho_r,x^\star(t)), \quad \forall t \geq 0,
\end{equation}
and is uniformly ultimately bounded as
\begin{equation}\label{eqn:ref:ultimate_tube}
    x_r(t) \in \Omega(\mu(\omega,T),x^\star(t))~ \subset~ \Omega(\rho_r,x^\star(t)), \quad \forall t \geq T > 0,
\end{equation} where the ultimate bound is defined as
\begin{equation}\label{eqn:ref:ultimate_bound}
    \mu(\omega, T) := \sqrt{\frac{e^{-2\lambda T}\mathcal{E}(x^\star_0,x_0)}{\underline{\alpha}} + \zeta_1(\omega)}.
\end{equation}
\end{lemma}

Next, we compute the bounds between the reference system in~\cref{eqn:reference_system} and the $\mathcal{L}_1$ closed-loop system (\cref{eqn:proset:dynamics} with~\cref{eqn:control_def:overall_control}).
\begin{lemma}
\label{lem:realbound}
Suppose that the stated assumptions and the conditions in~\cref{eqn:bandwidth_rate_conditions} hold. Additionally, assume that the trajectory of the $\mathcal{L}_1$ closed-loop system satisfies $x(t) \in \Omega(\rho,x^\star(t))$, for all $t \in [0, \tau]$, for some $\tau > 0$, with $\Omega(\rho,x^\star(t))$ and $\rho$ defined in~\cref{eqn:omega_norm_ball} and~\cref{eqn:conditions:rho}, respectively.
Then,
\[
\norm{x_r - x}_{\mathcal{L}_\infty}^{[0,\tau]} < \rho_a,
\] where $\rho_a$ is given in~\cref{eqn:conditions:rho}.
\end{lemma}

\begin{figure}
  \centering
    \includegraphics[width=0.8\textwidth]{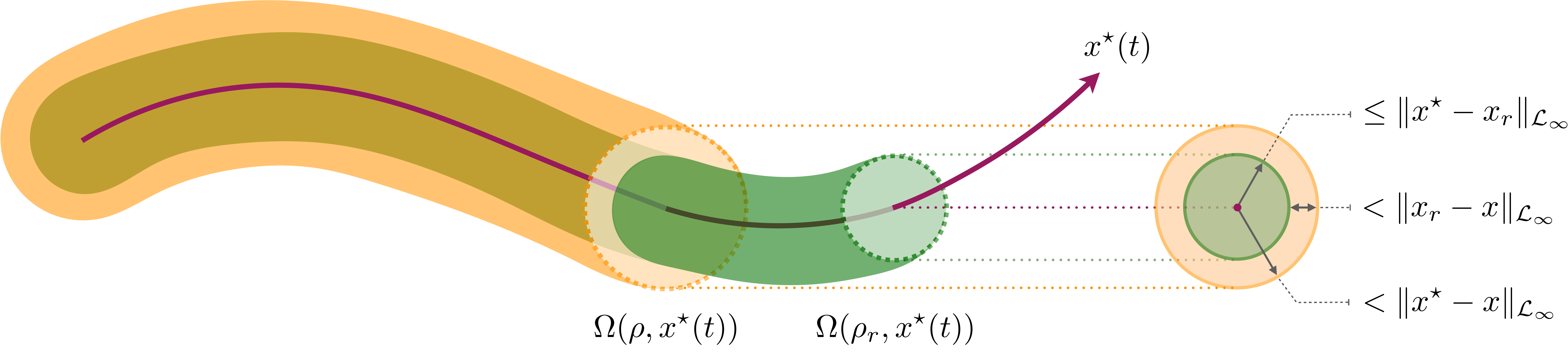}
  \caption{The bounds/tubes for the analysis of the CCM based $\mathcal{L}_1$-adaptive controller.}
  \label{fig:analysis_bounds}
\end{figure}

We now use Lemmas~\ref{lem:refbound}-\ref{lem:realbound} to state the main result of the paper.
\begin{theorem}\label{thm:main_theorem}
Suppose that the stated assumptions and conditions in~\cref{eqn:bandwidth_rate_conditions} hold. Consider a desired state trajectory~$x^\star(t)$ as in~\eqref{eqn:proset:desired_trajectory} and the state of the $\mathcal{L}_1$ closed-loop system defined via~\eqref{eqn:proset:dynamics} and~\eqref{eqn:control_def:overall_control}. Then we have
\begin{equation}\label{eqn:main:uniform_tube}
    x(t) \in \Omega (\rho,x^\star(t)), \quad \forall t \geq 0,
\end{equation} and is uniformly ultimately bounded as
\begin{equation}\label{eqn:main:ultimate_tube}
    x(t) \in \Omega (\delta(\omega,T),x^\star(t))~ \subset ~ \Omega (\rho,x^\star(t)), \quad \forall t \geq T > 0.
\end{equation} Here, the ultimate bound is defined as
\begin{equation}\label{eqn:main:ultimate_bound}
    \delta(\omega, T) := \mu(\omega, T) + \rho_a,
\end{equation}  where the positive scalars $\rho$ and $\rho_a$ are defined in~\eqref{eqn:conditions:rho}, and $\mu(\omega, T)$ is defined in~\cref{lem:refbound}.
\end{theorem}

\subsection*{Discussion}

A few critical comments are in order for the performance analysis. The main result in Theorem~\ref{thm:main_theorem} provides  uniform  ultimate bounds. Let us first discuss the implication of the uniform bound $\rho$ in~\cref{eqn:main:uniform_tube}. As per the definition in~\cref{eqn:conditions:rho}, $\rho = \rho_r + \rho_a$. It is evident from the definition that $\rho$ is lower bounded by the initial condition difference $\norm{x^\star_0 - x_0}$ and the positive scalars $\underline{\alpha}$ and $\overline{\alpha}$ which are associated with the CCM $M(x)$ of the nominal dynamics. Furthermore, as per the proof of Lemma~\ref{lem:realbound}, since $\rho_a \propto 1/\sqrt{\Gamma}$, the adaptation rate $\Gamma$ can be increased to the maximum value allowable by the computation hardware to guarantee the smallest $\rho_a$, and thus, the smallest uniform bound $\rho$. However, the fact remains that the uniform bound $\rho$ guaranteed by the $\mathcal{L}_1$-controller for the tracking remains lower bounded by $\norm{x^\star_0 - x_0}$. The only way this bound can be further reduced is if the underlying planner which provides the desired state-input pair $(x^\star(t),u^\star(t))$ can minimize $\norm{x^\star_0 - x_0}$.

 Theorem~\ref{thm:main_theorem} also provides the (uniform) ultimate bound via $\delta(\omega,T)$ defined in \cref{eqn:main:ultimate_bound}. As already mentioned, $\rho_a \propto 1 / \sqrt{\Gamma}$. Furthermore, from the definition of $\zeta_1(\omega)$ in~\cref{eqn:conditions:zeta_1}, it is evident that by choosing a large enough $\omega$, there will always exist a known $0 < T < \infty$ such that $\delta(\omega,t) \leq \bar{\rho}$, for all $t \geq T$, for any chosen $\bar{\delta}>0$. Therefore, we can always arbitrarily shrink the tube $\mathcal{O}(\bar{\delta})$ by choosing appropriate bandwidth $\omega$ and rate of adaptation $\Gamma$. This feature of the CCM-based $\mathcal{L}_1$-controller is very advantageous, since, for example, this capability will allow the safe navigation of a robot through tight and cluttered environments. This improved performance, however, comes at the cost of reduced robustness. There exists a trade-off between performance and robustness that should be taken into consideration. As aforementioned, performance (radius of tubes around $x^\star(t)$) depends on $\Gamma$ and $\omega$. The rate of adaptation $\Gamma$ is obviously limited by the available computational hardware. More importantly, the role of the low-pass filter $C(s)$ in the $\mathcal{L}_1$-control architecture (\cref{fig:arch}) is to decouple the control loop from the estimation loop~\cite{hovakimyan2010L1}. Thus, increasing the bandwidth $\omega$ of $C(s)$ in order to get a tighter tube will lead to the $u_a(t)$ component of the $\mathcal{L}_1$-input to behave as a high-gain signal, thus possibly sacrificing desired robustness levels~\cite{cao2010stability}. Therefore, this trade-off must always be taken into account during the planning phase.


\section{Simulation Results}\label{sec:sim}
We provide two illustrative examples. In the first example, we consider the non-feedback linearizable system from \cite{manchester2017control} and synthesize the controller to ensure safe regulation around the equilibrium point.  We also show the effect of uniform ultimate bounds discussed in the previous section, if the system were to start far away from the equilibrium. In the second example, we consider the system from \cite{singh2017robust} and ensure safety in a motion planning context during trajectory tracking. In particular we show how altering the tube parameters affects the choice in the filter bandwidth and adaptation rate.

\subsection{Non-feedback Linearizable Systems}
Consider the system with the structure defined in \cref{eqn:proset:dynamics} and the system functions given by
\[
f(x) = \begin{bmatrix}
-x_1(t) + x_3(t) \\
x_1^2(t) - 2x_1(t)x_3(t) - x_2(t)  + x_3(t) \\
-x_2(t)
\end{bmatrix},
B = \begin{bmatrix}
0 \\ 0 \\ 1
\end{bmatrix},
\]
where the state $x(t) = [x_1(t) \ x_2(t) \ x_3(t)]^\top$. The dual metric $W(x) = M(x)^{-1}$ satisfying the conditions in \cref{eqn:prelim:CCM_1,eqn:prelim:CCM_2,eqn:prelim:CCM_3} was found using the sum-of-squares programming toolbox SumOfSquares.jl \cite{julia2020sos}, optimization software JuMP \cite{dunning2017jump}, and the optimization solver \cite{julia2020mosek}, as
\[
W(x) = 
\begin{bmatrix}
0.2 & -0.41x_1(t) & -0.01\\
-0.41x_1(t) & 0.81x_1(t)^2 + 0.22 & 0.01x_1(t) - 0.01\\
-0.01 & 0.01x_1(t) - 0.01 & 0.07x_1(t) + 0.22
\end{bmatrix}.
\]
The metric satisfies a convergence rate $\lambda = 1.0$ and is uniformly bounded in the set $\mathcal{X} = \{y \in \mathbb{R}^3 \ | \ \norm{y}_\infty \le 0.1\}$ with $\oalpha = 5.88$ and $\ualpha = 3.85$. Now, suppose that the system is experiencing sinusoidal disturbances of the form: $h(t) = 0.1\sin(2t)$. We chose the initial condition of the system as $x_0 = [1 \ -1\ 1]^\top \times 10^{-2}$ and the desired state as $x^\star = [0 \ 0 \ 0]^\top$. Incidentally, the desired state is also the equilibrium point of the system which means that the desired control is $u^\star(t) \equiv 0$. 
\begin{figure}[ht]
    \centering
    \subfloat[]{\label{fig:ex1ccm}\includegraphics[width=0.5\textwidth]{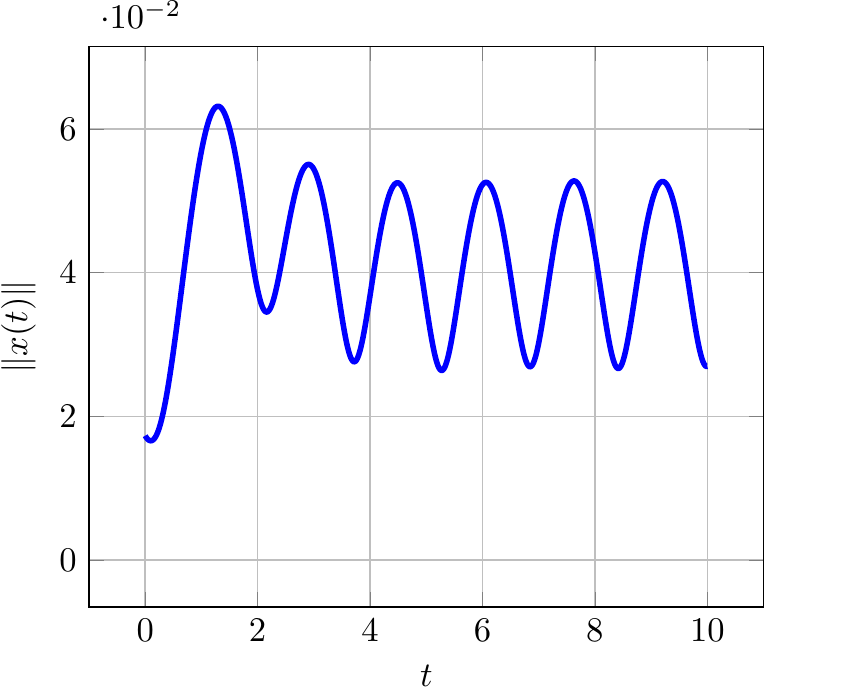}}
    \subfloat[]{\label{fig:ex1l1}\includegraphics[width=0.5\textwidth]{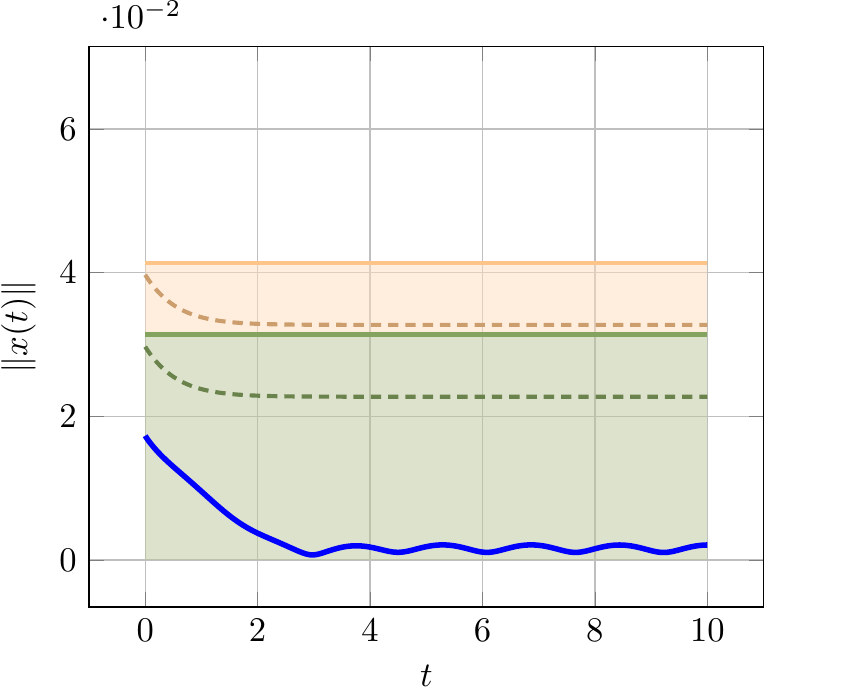}}
    \caption{Comparison of controller performance between (a) pure CCM-based feedback, and a (b) CCM-based $\mathcal{L}_1$ architecture. The green and orange shaded regions signify the induced $\Omega(\rho_r, x^\star)$ and $\Omega(\rho,x^\star)$ tubes respectively. The dashed green and orange lines signify the uniform ultimate bounds $\mu(\omega, T)$ and $\delta(\omega, T)$ evaluated at every timestep.}
    \label{fig:example1}
\end{figure}
A pure CCM-based feedback strategy produces the oscillatory behavior, seen in \cref{fig:ex1ccm}. A CCM-based $\mathcal{L}_1$-adaptive controller is designed in \cref{fig:ex1l1} for tube widths $\epsilon = 0.01$ and $\rho_a = 0.01$. The filter bandwidth and adaptation rate required to achieve this level of performance were chosen as $\omega = 50$ and $\Gamma = 5 \times 10^6$ respectively by satisfying the conditions in \cref{eqn:bandwidth_rate_conditions}. Notice that the bounds are far more conservative than the actual behavior of the system. In fact, the error in tracking is uniformly bounded as $\norm{x}_{\mathcal{L}_\infty} < 0.02$. Additionally, notice that the uniform ultimate bounds of the reference system tube from \cref{eqn:ref:ultimate_bound} and the actual system tube from \cref{eqn:main:ultimate_bound} shrink with time and are essentially `forgetting' the initial conditions of the system.

\subsection{Safe Tubes for Motion Planning}
Consider the system with the structure defined in \cref{eqn:proset:dynamics} and the system functions given by
\[
f(x) = \begin{bmatrix}
-x_1(t) + 2x_2(t) \\
-0.25x_2^3(t) -3x_1(t) + 4x_2(t)
\end{bmatrix}, \quad
B = \begin{bmatrix}
0.5 \\ -2
\end{bmatrix},
\]
where the state $x(t) = [x_1(t) \ x_2(t)]^\top$. Since this particular system is feedback linearizable, it admits a constant (or flat) dual metric for all $x \in \mathbb{R}^2$ \cite{manchester2018unifying}. The value of the dual metric and the associated convergence parameter is computed in \cite{singh2017robust} and provided here for completeness:
\[
W = \begin{bmatrix}
4.26 & -0.93\\
-0.93 & 3.77
\end{bmatrix},  \quad \lambda = 1.74.
\]
Similar to \cite{singh2017robust}, we chose the initial condition of the system as $x_0 = [3.4 \ -2.4]^\top$ and the target state as as $x^\star = [0 \ 0]^\top$. The desired state and control trajectory pair was computed using the iterative LQR solver provided by \cite{howell2020traj} with the parameters $Q = 0.5\mathbb{I}_2$ and $R = 1.0$. Suppose the system is affected by uncertainties of the form: $h(t,x) = - 2\sin(2t) - 0.1\norm{x(t)}$, consisting of both time and state dependent terms. Depending on the desired level of tracking performance or closeness to obstacles in the environment, the user will pick the tube parameters $\epsilon$ and $\rho_a$ as defined in \cref{eqn:conditions:rho}.  In \cref{fig:example2-bounds}, we illustrate the trade-offs between choosing a tighter $\rho_a$ (\cref{fig:ex2left}) versus a tighter $\epsilon$ (\cref{fig:ex2right}) for this system.

\begin{figure}[ht]
    \centering
    \subfloat[][$\epsilon=0.6,\rho_a=0.01$,\\$\omega=51, \Gamma=2\times10^9$]{\label{fig:ex2left}\includegraphics[width=0.33\textwidth]{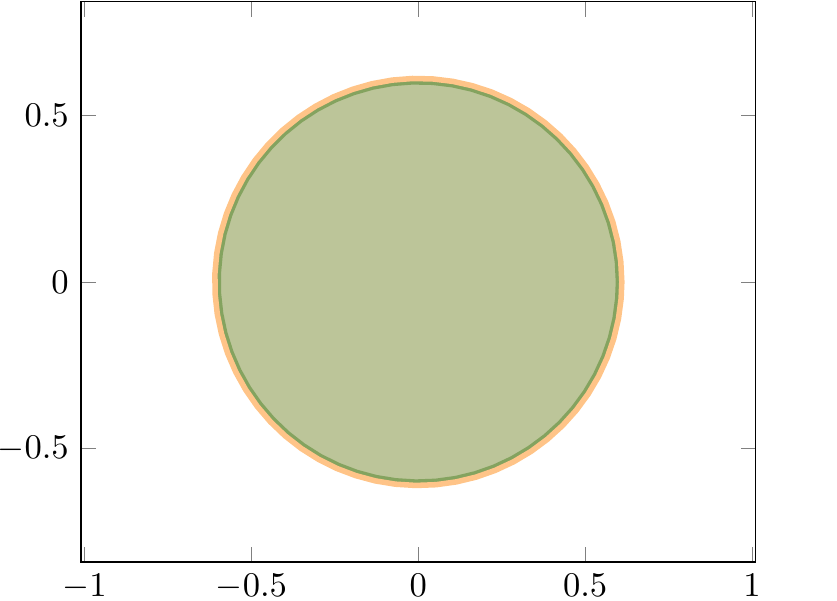}} \subfloat[(b)][$\epsilon=0.6,\rho_a=0.1$,\\$\omega=60, \Gamma=2\times10^7$]{\label{fig:ex2middle}\includegraphics[width=0.33\textwidth]{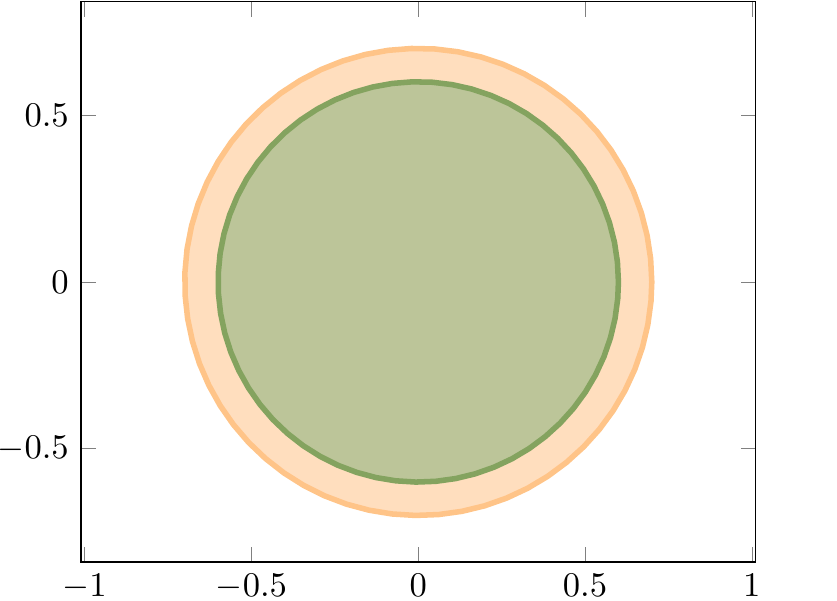}} 
    \subfloat[(c)][$\epsilon=0.4,\rho_a=0.1$,\\$\omega=90, \Gamma=4\times10^7$]{\label{fig:ex2right}\includegraphics[width=0.33\textwidth]{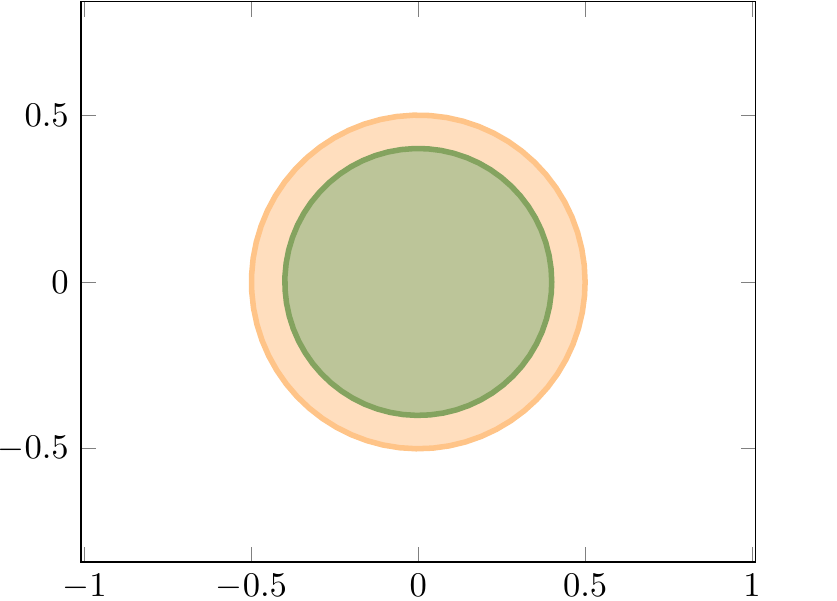}}
    \caption{Relationship between the choice of tube parameters $\epsilon$ and $\rho_a$  and the controller parameters $\omega$ and $\Gamma$ through the conditions defined in \cref{eqn:bandwidth_rate_conditions}. For clarity the initial conditions for the desired trajectory and the actual system in this illustration are assumed to be the same.}
    \label{fig:example2-bounds}
\end{figure}

In \cref{fig:ex2l1}, we observe the performance and robustness benefits of using CCM-based $\mathcal{L}_1$-adaptive control. Not only does the system track the desired trajectory closely, but also avoids colliding with obstacles (unlike in \cref{fig:ex2ccm}) through an appropriate choice of tube parameters. 

\begin{figure}[ht]
    \centering
    \subfloat[]{\label{fig:ex2ccm}\includegraphics[width=0.5\textwidth]{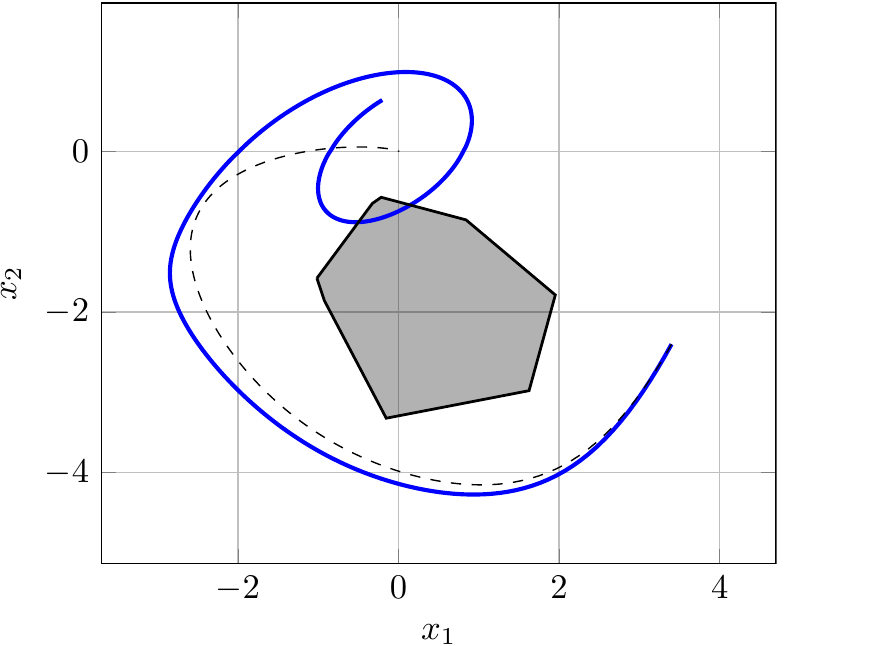}}
    \subfloat[]{\label{fig:ex2l1}\includegraphics[width=0.5\textwidth]{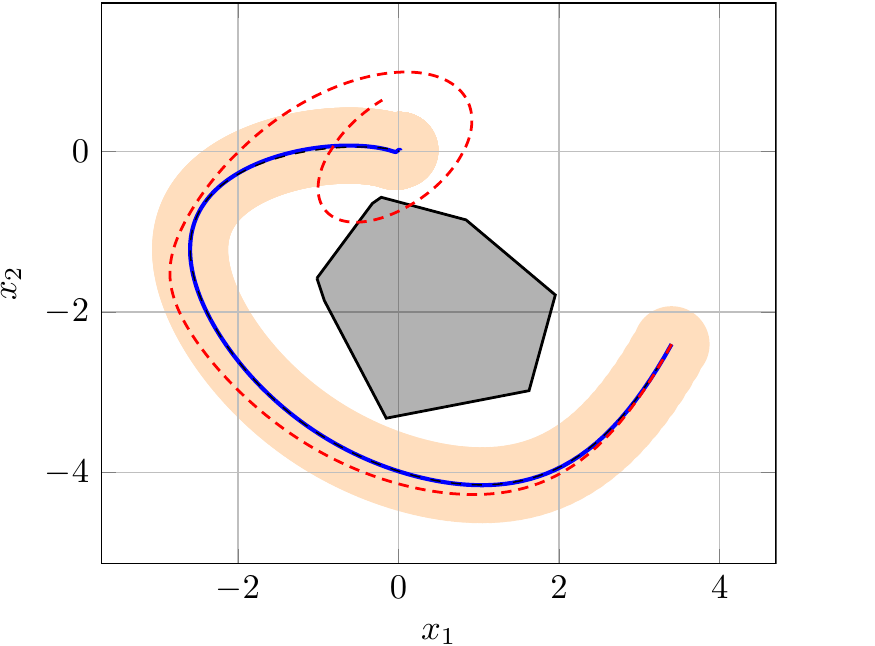}}
    \caption{Comparison of performance and robustness between (a) pure CCM-based feedback, and (a) CCM-based $\mathcal{L}_1$ architecture with tube parameters $\epsilon = 0.4$ and $\rho_a = 0.1$. The dashed black line shows the desired trajectory designed by a planner, the gray polygon is an obstacle, and the orange shaded region is the safe tube given by $\Omega(\rho,x^\star(t))$. The behavior of the system under pure CCM-based feedback has been overlaid as a dashed red line in (b) for clarity.}
    \label{fig:example2}
\end{figure}


\section{Conclusion and Future Work}\label{sec:numexp}
We present a control methodology to enable safe and guaranteed feedback motion planning. The presented work relies on differential geometric contraction theory and $\mathcal{L}_1$-adaptive control. The proposed controller enables the apriori computation of uniform and ultimate-bounds which act as safety-certificates. These safety certificates induce `tubes' which can be taken into account by any planner of choice. In this way, the safety of the system/robot is always guaranteed in the presence of model and environmental uncertainties. Furthermore, by using the control law's filter bandwidth and rate of adaptation as tuning knobs, the width of the safety tubes can be adjusted.
Future research will deal with the incorporation of learning into the proposed control framework. This will lead to improved performance since the underlying planner will have access to improved model knowledge. More importantly, the controller proposed in this work will keep the learning process safe because of the apriori specified and guaranteed bounds during the transient phase. Further investigations will be undertaken to extend this framework to enable safe planning and control using a reduced-order model to enable fast planning.


\section*{Acknowledgements}
This work is financially supported by Air Force Office of Scientific Research (AFOSR), National Aeronautics and Space Administration (NASA) and National Science Foundation’s National Robotics Institute (NRI) and Cyber Physical Systems (CPS) awards \#1830639 and \#1932529. We would also like to thank Dr.~Riccardo Bonalli for fruitful discussions regarding the proof for \cref{lem:roadblock}, Dr.~Evangelos Theodorou for introducing us to world of stochastic contraction theory \cite{bouvrie2019wasserstein} which will be especially useful in future work, and Dr.~Xiaofeng Wang for his valuable comments during the preparation of this manuscript.

\renewcommand{\thesection}{\Alph{section}}
\renewcommand{\theequation}{\Alph{section}.\arabic{equation}}
\setcounter{section}{0}
\setcounter{equation}{0}
\section*{\LARGE Appendix}
\section{Technical Results}
\begin{lemma}
\label{lem:scalar}
Consider a scalar system with vector inputs
\begin{align*}
\dot{z}(t) &= -az(t) + b(t)^\top\xi(t), \quad z(0) = 0, \\
\xi(s) &= (\mathbb{I}_m - C(s))\sigma(s),
\end{align*}
where $z(t) \in \mathbb{R}$ is the state, $\sigma(t) \in \mathbb{R}^m$ is the control input with a column vector of transfer functions $\sigma(s)$, $a > 0$ is a scalar, $b(t) \in \mathbb{R}^m$ is differentiable, and $C(s)$ is a low-pass filter of the form $\frac{\omega}{s+\omega} \mathbb{I}_m$ with bandwidth $\omega > 0$. If the following bounds hold
\[
\quad \norm{b}_{\mathcal{L}_\infty}^{[0,\tau]} \le \Delta_b, \quad \norm{\dot{\sigma}}_{\mathcal{L}_\infty}^{[0,\tau]} \le \Delta_{\sigma_t},
\]
for some $\tau>0$, then the following inequality holds
\[
\norm{z}_{\mathcal{L}_\infty}^{[0,\tau]} \le  \Delta_b \left(\frac{\norm{\sigma(0)}}{\abs{a - \omega}}  + \frac{\Delta_{\sigma_t}}{a\omega} \right).
\]
\end{lemma}
\begin{proof}
The scalar system can be equivalently written as
\begin{align}
\label{eqn:app:scalar1}
\dot{z}(t) &= -az(t) + b(t)^\top\xi(t)  &&z(0) = 0  \\
\label{eqn:app:scalar2}
\dot{\xi}(t) &= -\omega \xi(t) + \dot{\sigma}(t), &&\xi(0) = \sigma(0).
\end{align}
The solution to the differential equation in \cref{eqn:app:scalar2} is
\[
\xi(\lambda) = e^{-\omega \lambda} \sigma(0) + \int_{0}^\lambda e^{-\omega (\lambda - \nu)} \dot{\sigma}(\nu) \diff \nu,
\]
and the solution to \cref{eqn:app:scalar1} is
\[
z(t) = \int_{0}^t e^{-a (t - \lambda)} b(\lambda)^\top \xi(\lambda) \diff \lambda.
\]
Combining the equations above, we have
\begin{align*}    \norm{z(t)} \le  \Delta_b \frac{\norm{\sigma(0)}(e^{-\omega t} - e^{-a t})}{a - \omega} + \Delta_b \frac{\Delta_{\sigma_t}
    }{\omega}\left( \frac{1 - e^{-a t}}{a} - \frac{e^{-\omega t} - e^{-at}}{a - \omega}\right), \quad \forall t \in [0,\tau].
\end{align*}
This expression can be further bounded as
\begin{align*}
    \norm{z(t)} \le  \Delta_b \left(\frac{\norm{\sigma(0)}}{\abs{a - \omega}}  + \frac{\Delta_{\sigma_t}}{a \omega}\right),
\end{align*}
for all $t \in [0, \tau]$.
\end{proof}

\begin{lemma}
\label{lem:scalar2}
Consider a scalar system with vector inputs
\begin{align*}
\dot{z}(t) &= -az(t) + b(t)^\top\xi(t), \quad z(0) = 0, \\
\xi(s) &= (\mathbb{I}_m - C(s))\sigma(s),
\end{align*}
where $z(t) \in \mathbb{R}$ is the state,  $\sigma(t) \in \mathbb{R}^m$ is the control input with a column vector of transfer functions $\sigma(s)$, $a > 0$ is a scalar, $b(t) \in \mathbb{R}^m$ is differentiable, and $C(s)$ is a low-pass filter of the form $\frac{\omega}{s+\omega} \mathbb{I}_m$ with bandwidth $\omega > 0$. If the following bounds hold
\[
\quad \norm{b}_{\mathcal{L}_\infty}^{[0,\tau]} \le \Delta_b,  \quad \|\dot{b}\|_{\mathcal{L}_\infty}^{[0,\tau]} \le \Delta_{b_t}, \quad \norm{\sigma}_{\mathcal{L}_\infty}^{[0,\tau]} \le \Delta_{\sigma},
\]
for some $\tau>0$, then the following inequality holds
\[
\norm{z}_{\mathcal{L}_\infty}^{[0,\tau]} \le  \Delta_\sigma \frac{2 a \Delta_b + \Delta_{b_t}}{a\omega}.
\]
\end{lemma}
\begin{proof}
We closely follow the analysis in \cite[Lemma 1]{wang2017adaptive}, but show that the inequality holds for vector inputs as well. Since the input is bounded in truncation, the norm of the system solution is bounded as
\[
\norm{z}_{\mathcal{L}_\infty}^{[0, \tau]}  \le \norm{\mathcal{Y}}_{\mathcal{L}_1}^{[0, \tau]} \norm{\sigma}_{\mathcal{L}_\infty}^{[0, \tau]},
\]
where $\mathcal{Y}$ is the equivalent of a transfer function between $\sigma$ and $z$ \cite[Lemma A.7.1]{hovakimyan2010L1}. The impulse response $y(t)$ of $\mathcal{Y}$ to the Dirac delta function $\delta(t)$ is characterized as
\begin{align*}
    y(t) &= \int_0^t e^{-a(t - \lambda)}b(\lambda)^\top \int_0^\lambda \mathds{1}_m (\delta(\nu) - \omega e^{-\omega \nu}) \delta(\lambda - \nu) \diff\nu \diff\lambda \\
    &= \int_0^t e^{-a(t - \lambda)}b(\lambda)^\top \mathds{1}_m (\delta(\lambda) - \omega e^{-\omega \lambda}) \diff\lambda \\
    &= e^{-at} b(0)^\top \mathds{1}_m -  \int_0^t e^{-a(t - \lambda)}b(\lambda)^\top \mathds{1}_m \omega e^{-\omega \lambda} \diff\lambda.
\end{align*}
Integrating the second term by parts we obtain the solution of the system as
\begin{align*}
    y(t)&= e^{-\omega t} b(t)^\top \mathds{1}_m - \int_0^t \left(e^{-a(t - \lambda)}\dot{b}(\lambda) -ae^{-a(t - \lambda)}b(\lambda)\right)^\top \mathds{1}_m e^{-\omega \lambda} \diff\lambda.
\end{align*}
The $\mathcal{L}_1$ norm of $\mathcal{Y}$ is simply the norm of the impulse response, which is given by
\begin{align*}
    \mathcal{Y}_{\mathcal{L}_1}^{[0, \tau]} &= \norm{y}_{\mathcal{L}_1}^{[0, \tau]} = \int_0^\tau \norm{y(t)} \diff t \\
    &= \int_0^\tau \norm{e^{-\omega t} b(t)^\top \mathds{1}_m - \int_0^t \left(e^{-a(t - \lambda)}\dot{b}(\lambda) -ae^{-a(t - \lambda)}b(\lambda)\right)^\top \mathds{1}_m e^{-\omega \lambda} \diff \lambda}\diff t\\
    &\le \Delta_b \frac{1 - e^{-\omega \tau}}{\omega} + \int_0^\tau\norm{\int_0^t \left(e^{-a(t - \lambda)}\dot{b}(\lambda) -ae^{-a(t - \lambda)}b(\lambda)\right)^\top \mathds{1}_m e^{-\omega \lambda} \diff \lambda} \diff t \\
    &\le \Delta_b \frac{1 - e^{-\omega \tau}}{\omega} + (\Delta_{b_t} + a\Delta_b) \int_0^\tau\int_0^t e^{-a(t - \lambda)} e^{-\omega \lambda} \diff \lambda \diff t \\
    &\le \Delta_b \frac{1 - e^{-\omega \tau}}{\omega} + (\Delta_{b_t} + a\Delta_b) \left( \frac{1}{a\omega} - \frac{a e^{-\omega \tau} - \omega e^{-a\tau}}{a\omega(a-\omega)} \right).
\end{align*}
This expression can be further bounded as
\begin{align*}
    \mathcal{Y}_{\mathcal{L}_1}^{[0, \tau]} &\le \frac{2a\Delta_{b} + \Delta_{b_t}}{a\omega}
\end{align*}
Therefore, the solution of the system $z(t)$ is bounded as
\begin{align*}
    z(t) &\le \Delta_\sigma \frac{2a\Delta_{b} + \Delta_{b_t}}{a\omega},
\end{align*}
for all $t \in [0, \tau]$.
\end{proof}

\begin{lemma}
\label{lem:riemannbounds}
Consider the Riemannian manifold $(\mathcal{X},M)$ where the metric satisfies \cref{assmp:CCM}, i.e the following uniform bounds hold
\[
\ualpha \mathbb{I}_n \preceq M(x) \preceq \oalpha \mathbb{I}_n, \quad \forall x \in \mathcal{X},
\]
where $\oalpha \ge \ualpha > 0$. Then the Riemannian energy satisfies the following uniform bounds
\[
 \underline{\alpha} \norm{p - q}^2 \leq \mathcal{E}(p,q) \leq \oalpha \norm{p - q}^2, \quad \forall p,q \in \mathcal{X}.
\] 
\end{lemma}
\begin{proof}
Let $\Xi(p,q)$ be the set of smooth curves connecting $p,q \in \mathcal{X}$. Applying the Cauchy-Schwarz inequality to the definition of the length of curve given in \cref{eqn:prelim:length} produces
\begin{equation}\label{eqn:length_squared_UB1}
    l(\gamma)^2 \leq \int_0^1 \gamma_s(s)M(\gamma(s))\gamma_s(s)\diff s \triangleq E(\gamma), \quad \forall \gamma \in \Xi(p,q).
\end{equation} Furthermore, since the minimizing geodesic $\ogamma \in \Xi(p,q)$ has the property that
$\sqrt{\gamma_s(s)M(\gamma(s))\gamma_s(s)}$ is a constant for all $s \in [0,1]$~\cite[Lemma~5.5]{lee2006riemannian}, the application of Cauchy-Schwarz also produces
\begin{equation}\label{eqn:length_squared_UB2}
    l(\ogamma)^2 = \int_0^1 \ogamma_s(s)M(\ogamma(s))\ogamma_s(s)\diff s = E(\ogamma).
\end{equation} Since, by definition $\ogamma$ is a minimizer of $l(\gamma)$ and $l(\gamma)\ge 0$ always,  it is also a minimizer of $l^2(\gamma)$. Using Eqs.~\eqref{eqn:length_squared_UB1}-\eqref{eqn:length_squared_UB2}, we conclude that
\[
\inf_{\gamma \in \Xi(p,q)} l(\gamma)^2 = l(\ogamma)^2 = E(\ogamma) = \inf_{\gamma \in \Xi(p,q)} E(\gamma). 
\] By the definition of the Riemannian energy given in \cref{eqn:prelim:energy}, the following chain of equalities are satisfied
\[
\mathcal{E}(p,q) = d(p,q)^2 = \left(\inf_{\gamma \in \Xi(p,q)} l(\gamma)  \right)^2
= \inf_{\gamma \in \Xi(p,q)}  l(\gamma)^2.
\] Using the previous two equalities, we conclude that
\begin{equation}\label{eqn:Riemannian_enrgy_minimizer_energy}
\mathcal{E}(p,q) = \inf_{\gamma \in \Xi(p,q)} E(\gamma).    
\end{equation} Using the bounds on $M(x)$, we further obtain
\begin{equation}\label{eqn:Riemannian_energy_int_1}
\ualpha \inf_{\gamma \in \Xi(p,q)} \int_0^1 \gamma_s(s)^\top \gamma_s(s) \diff s \leq   \mathcal{E}(p,q) \leq \oalpha \inf_{\gamma \in \Xi(p,q)} \int_0^1 \gamma_s(s)^\top \gamma_s(s) \diff s.
\end{equation} The expression $\inf_{\gamma \in \Xi(p,q)} \int_0^1 \gamma_s(s)^\top \gamma_s(s) \diff s$ represents the Riemannian energy over a manifold with the flat (i.e. constant) metric $\mathbb{I}_n$. Under the flat metrics, geodesics are straight lines, i.e., $\gamma(s) = (1-s)p + sq$, $s \in [0,1]$. Thus,
\[
\inf_{\gamma \in \Xi(p,q)} \int_0^1 \gamma_s(s)^\top \gamma_s(s)\diff s = \norm{p -q}^2.
\] Substituting into~\eqref{eqn:Riemannian_energy_int_1} then completes the result.
\end{proof}

\begin{lemma}
\label{lem:geobound}
Consider a minimizing geodesic $\ogamma: [0,1] \to \mathcal{X}$ under the metric $M(x)$, $x \in \mathcal{X}$, between points $p,q \in \mathcal{X}$ such that $\ogamma(0) = p$ and $\ogamma(1) = q$. If the metric $M(x)$ satisfies Assumption~\ref{assmp:CCM}, then the following inequalities are satisfied
\begin{align}
    \label{eqn:app:geobound1}\norm{M(\ogamma(s))\ogamma_s(s)} &\le \oalpha \norm{p - q}, \\
    \label{eqn:app:geobound2}\norm{\ogamma_s(s)} &\le \sqrt{\frac{\oalpha}{\ualpha}} \norm{p - q},
\end{align}
for all $s \in [0, 1]$,  where $\Theta(x)^\top \Theta(x) = M(x)$. 
\end{lemma}
\begin{proof}
From the proof of Lemma~\ref{lem:riemannbounds}, the Riemannian distance $d(p,q)$ is bounded as
\[
\sqrt{\ualpha} \norm{p - q} \le d(p, q) \le \sqrt{\oalpha} \norm{p - q},
\]
since the metric $M(x)$ satisfies Assumption~\ref{assmp:CCM}. Geodesics exhibit the special property \cite[Lemma 5.5]{lee2006riemannian} that $\overline{\gamma}_s(s)^\top M(\overline{\gamma}_s(s)) \overline{\gamma}_s(s) = c > 0$ for all $s \in [0, 1]$, i.e. they are of constant speed. Since $s \in [0, 1]$, using the proof of Lemma~\ref{lem:riemannbounds}, it can be shown that the speed of a geodesic is also the length of the geodesic
\begin{equation}
\label{eqn:app:geobound3}
\sqrt{\overline{\gamma}_s(s)^\top M(\overline{\gamma}(s)) \overline{\gamma}_s(s)} = d(p, q) \le \sqrt{\oalpha} \norm{p - q},
\end{equation}
for all $s \in [0, 1]$. Therefore given the factorization $M(x) = \Theta(x)^\top \Theta(x)$, we obtain the following inequality
\[
\norm{\Theta(\ogamma(s))\ogamma_s(s)} \le \sqrt{\oalpha} \norm{p - q}.
\]
Since $\norm{M(\ogamma(s))\ogamma_s(s)} \le \norm{\Theta(\ogamma(s))^\top}\norm{\Theta(\ogamma(s))\ogamma_s(s)}$ (because of the submultiplicative property of induced matrix norms and $\norm{\Theta(x)} \le \sqrt{\alpha}$ for all $x \in \mathcal{X}$), we arrive at \cref{eqn:app:geobound1}. Additionally, since $M(x)$ is uniformly bounded, the following inequality holds
\[
\ualpha \ogamma_s(s)^\top \mathbb{I}_n \ogamma_s(s) \le
\ogamma_s(s)^\top M(\ogamma(s)) \ogamma_s(s),
\]
for all $s \in [0, 1]$. Combining the above inequality with \cref{eqn:app:geobound3}, we obtain the result in \cref{eqn:app:geobound2}.
\end{proof}

\begin{lemma}
\label{lem:roadblock}
Consider two smooth trajectories $x_0$ and $x_1$ such that $x_0(t),x_1(t) \in \Omega(\rho, x^\star(t))$, for all $t \in [0, \tau]$, for some $\tau>0$, and a minimizing geodesic $\overline{\gamma}(\cdot, t): [0, 1] \to \mathcal{X}$ under the metric $M(x)$ such that $\overline{\gamma}(0, t) = x_0(t)$ and $\overline{\gamma}(1, t) = x_1(t)$. Then, if $M(x)$ satisfies Assumption~\ref{assmp:CCM}, we have
\[
\norm{\ogamma_s(1,t)^\top M(\ogamma(1,t))-\ogamma_s(0,t)^\top M(\ogamma(0,t))}  \le \frac{\overline{\alpha}}{2\ualpha}\Delta_{M_x}\norm{x_0(t) - x_1(t)}^2, \quad  \forall t \in [0, \tau],
\]
where $\Delta_{M_x}$ is defined in \cref{eqn:bounds:M_x}.
\end{lemma}
\begin{proof}
Recall that the \textit{differential Lyapunov function} is defined as $V(x, \delta x) := \delta x^\top M(x) \delta x$. Then, along the geodesic $\ogamma$, the differential Lyapunov function can be written as
\[
V(\ogamma(s,t),\ogamma_s(s,t)) = \ogamma_s(s,t)^\top M(\ogamma(s,t))\ogamma_s(s,t), \quad s \in [0,1].
\] Thus, we obtain the expressions
\begin{align*}
    \pdv{V}{x}(x,\delta x) = & 
    \begin{bmatrix}  
    \delta x^\top \pdv{M}{x_1}(x)\delta x & \cdots & \delta x^\top \pdv{M}{x_n}(x)\delta x
    \end{bmatrix} := \delta_x^\top \pdv{M}{x}(x) \delta x \in \mathbb{R}^{1 \times n},\\
    \pdv{V}{\delta x}(x,\delta x) = & 2 \delta x^\top M(x) \in \mathbb{R}^{1 \times n},
\end{align*} which evaluated at the geodesic $\ogamma$ are given by
\begin{align}
    \pdv{V}{x}(\ogamma(s,t),\ogamma_s(s,t)) = & 
    \begin{bmatrix}  
    \ogamma_s(s,t)^\top \pdv{M}{x_1}(\ogamma(s,t))\ogamma_s(s,t) & \cdots & \ogamma_s(s,t)^\top \pdv{M}{x_n}(\ogamma(s,t))\ogamma_s(s,t) 
    \end{bmatrix} \notag \\
    = & \label{eqn:roadblock:V_diff_x} \ogamma_s(s,t)^\top \pdv{M}{x}(\ogamma(s,t)) \ogamma_s(s,t),\\
    \pdv{V}{\delta x}(\ogamma(s,t),\ogamma_s(s,t)) = & \label{eqn:roadblock:V_diff_delta_x} 2 \ogamma_s(s,t)^\top M(\ogamma(s,t)),
\end{align} for $s \in [0,1]$.

Now, $\ogamma(s,t)$ minimizes the energy functional $E(\gamma(s,t))$ defined in the proof of Lemma~\ref{lem:riemannbounds}, where $\gamma(\cdot,t)$ is any smooth curve connecting $x_0(t),x_1(t) \in \Omega(\rho,x^\star(t)) \subset \mathcal{X}$. In other words, $\ogamma(s,t)$ minimizes the functional
\[
\int_0^1 \gamma_s(s,t)^\top M(\gamma(s,t))\gamma_s(s,t)\diff s = \int_0^1 V(\ogamma(s,t),\ogamma_s(s,t))\diff s. 
\]
Therefore, $\ogamma(s,t)$ satisfies the Euler-Lagrange equations~\cite[Appendix D]{bishop2006pattern} given by
\begin{align*}
\frac{\diff}{\diff s} \pdv{V}{\delta_x} (\ogamma(s,t),\ogamma_s(s,t)) = \pdv{V}{x}(\ogamma(s,t),\ogamma_s(s,t)), \quad s \in [0,1].
\end{align*} 
Using~\eqref{eqn:roadblock:V_diff_x}-\eqref{eqn:roadblock:V_diff_delta_x}, we get
\[
\frac{\diff}{\diff s} \ogamma_s(s,t)^\top M(\ogamma(s,t))  = \frac{1}{2} \ogamma_s(s,t)^\top \pdv{M}{x}(\ogamma(s,t)) \ogamma_s(s,t), \quad s \in [0,1].
\] 
Integrating both sides and applying the fundamental theorem of calculus \cite[Section 5.3]{apostol1966calculus} produces
\[
\left[\ogamma_s(s,t)^\top M(\ogamma(s,t))\right]_{s=0}^{s=1}  = \frac{1}{2} \int_0^1 \ogamma_s(s,t)^\top \pdv{M}{x}(\ogamma(s,t)) \ogamma_s(s,t)\diff s.
\] 
Component-wise, this expression can be written as
\begin{equation}\label{eqn:roadblock:EL}
\left[\ogamma_s(s,t)^\top M(\ogamma(s,t))\right]_{s=0}^{s=1}[i]  = \frac{1}{2} \int_0^1 \ogamma_s(s,t)^\top \pdv{M}{x_i}(\ogamma(s,t)) \ogamma_s(s,t)\diff s, \quad i \in \{1,\dots,n\}.
\end{equation}
From \cref{eqn:bounds:M_x} and \cref{lem:geobound} the following bounds hold
\[
\sum_{i=1}^n \norm{\pdv{M}{x_i}(x)} \le \Delta_{M_x},
\qquad
\norm{\ogamma_s(s,t)}^2 \le \frac{\overline{\alpha}}{\ualpha} \norm{x_0 - x_1}^2
\]
for all $x(t) \in \Omega(\rho, x^\star(t)) \in \mathcal{X}$, $t \in [0, \tau]$ and $s \in [0,1]$. Since $\ogamma(s,t) \in \mathcal{X}$, for all $s \in [0,1]$, using the aforementioned bounds and~\eqref{eqn:roadblock:EL}, we 
arrive at the result.
\end{proof}

\begin{lemma}
\label{lem:dynbound}
Let the state $x_r(t)$ of the reference system in~\cref{eqn:reference_system} and the state $x(t)$ of the real system in~\cref{eqn:proset:dynamics} with control input~\cref{eqn:control_def:overall_control} satisfy 
$x_r(t), x(t) \in \Omega(\rho, x^\star(t))$ for all $t \in [0, \tau]$, for some $\tau>0$. Additionally, let Assumptions~\ref{assmp:desired_control}-\ref{assmp:bounds_uncertainty} and~\ref{assmp:CCM} hold. Then, the following inequalities are satisfied
\begin{align}
\label{eqn:app:dx_r}
\norm{\dot{x}_r}_{\mathcal{L}_\infty}^{[0,\tau]} \le \Delta_{\dot{x}_r} , \\
\label{eqn:app:dx}
\norm{\dot{x}}_{\mathcal{L}_\infty}^{[0,\tau]} \le \Delta_{\dot{x}},
\end{align}
where $\Delta_{\dot{x}_r}$ and $\Delta_{\dot{x}}$ are defined in \cref{eqn:bounds:dx_r,eqn:bounds:dx} respectively.
\end{lemma}
\begin{proof}
Using the dynamics of the reference system in~\cref{eqn:reference_system}, we obtain
\[
\norm{\dot{x}_r(t)} \leq \norm{f(x_r(t))} + \norm{B(x_r(t))}\left(\norm{u_{c,r}(t)} + \norm{\tilde{h}(t,x_r(t))}  \right), \quad \forall t \in [0,\tau],
\] 
where $\tilde{h}(t,x_r(t)) = h(t,x_r(t)) - \eta_r(t)$. Thus, using~\cref{assmp:bounds_known}, we obtain
\begin{equation}\label{eqn:B7:reference_system}
\norm{\dot{x}_r(t)} \leq \Delta_f + \Delta_B \left(\norm{u_{c,r}(t)} + \norm{\tilde{h}(t,x_r(t))}  \right), \quad \forall t \in [0,\tau].
\end{equation}
Using the definition of $u_{c,r}(t)$ in~\cref{eqn:reference_system} and~\cref{assmp:desired_control}, we get the following bound
\[
\norm{u_{c,r}(t)} \leq \Delta_{u^\star} + \norm{k(x^\star(t),x_r(t))}, \quad \forall t \in [0,\tau].
\] We can use the bound on the feedback term $\norm{k(x^\star(t),x_r(t))}$ as in \cite[Theorem 5.2 \& Eqn. (49)]{singh2019robust} to get
\begin{equation}\label{eqn:B7:reference_bound_1}
 \norm{u_{c,r}(t)} \leq \Delta_{u^\star} + \rho \Delta_{\delta_u}, \quad \forall t \in [0,\tau].    
\end{equation} 
Note that by definition $\mathscr{L}[\tilde{h}(t, x_r)] = \mathscr{L}[h(t, x_r) - \eta_r(t)] = (\mathbb{I}_m - C(s))\mathscr{L}[h(t, x_r)]$. Thus, using~\cite[Lemma A.7.1]{hovakimyan2010L1}, we have
\[
\norm{\tilde{h}(t,x_r(t))} = \norm{h(t,x_r(t)) - \eta_r(t)} \leq \norm{\mathbb{I}_m - C(s)}_{\mathcal{L}_1}\norm{h(\cdot,x_r(\cdot))}_{\mathcal{L_\infty}}^{[0,\tau]}, \quad \forall t \in [0,\tau].
\] 
Using the fact that $\norm{h(t,x_r)} \le \Delta_h$ for all $x_r \in \Omega(\rho, x^\star(t)$ and $t \in [0, \tau]$, from \cref{assmp:bounds_uncertainty} we obtain
\begin{equation}\label{eqn:B7:reference_bound_2}
\norm{\tilde{h}(t,x_r(t))} \leq \norm{\mathbb{I}_m - C(s)}_{\mathcal{L}_1}\Delta_{h}, \quad \forall t \in [0,\tau].    
\end{equation} Then, substituting~\eqref{eqn:B7:reference_bound_1}-\eqref{eqn:B7:reference_bound_2} into~\eqref{eqn:B7:reference_system} establishes the inequality in~\cref{eqn:app:dx_r}. Using the similar line of reasoning, we obtain the following bound on the state of the real system ~\cref{eqn:proset:dynamics} with control input~\cref{eqn:control_def:overall_control}
\begin{equation}\label{eqn:B7:real_system}
    \norm{\dot{x}(t)} \leq \Delta_f + \Delta_B \left(\Delta_h + \norm{u(t)}   \right), \quad \forall t \in [0,\tau].
\end{equation} Using the definition of $u_a(t)$ in~\cref{eqn:control_def:control_law} and again applying~\cite[Lemma A.7.1]{hovakimyan2010L1} we get
\[
\norm{u_a(t)} \leq \norm{C(s)}_{\mathcal{L}_1}\norm{\hat{\sigma}}_{\mathcal{L}_\infty}^{[0,\tau]}, \quad \forall t \in [0,\tau].
\] We have that $\norm{C(s)}_{\mathcal{L}_1} = 1$ for the chosen first-order low-pass filter. Additionally, since $\hat{\sigma}(t) \in \mathcal{H}$ due to the adaptation law defined in~\cref{eqn:control_def:adaptation_law}, we conclude that $\norm{\hat{\sigma}}_{\mathcal{L}_\infty}^{[0,\tau]} \leq \Delta_h$. Thus, $\norm{u_a(t)} \leq \Delta_h$ for all $t \in [0,\tau]$. And, since $x(t) \in \Omega(\rho, x^\star(t))$ for all $t \in [0, \tau]$, the upperbound on $\norm{u_c(t)}$ is equivalent to that of \cref{eqn:B7:reference_bound_1}. Therefore,
\begin{equation}
    \label{eqn:B7:controlbound}
    \norm{u(t)} \leq \norm{u_c(t)} + \norm{u_a(t)} \leq \Delta_u^\star + \rho \Delta_{\delta_u} + \Delta_h, \quad \forall t \in [0, \tau].
\end{equation}
Substituting into~\cref{eqn:B7:real_system} then establishes the inequality in~\cref{eqn:app:dx}.
\end{proof}
\begin{lemma}
\label{lem:tildebound}
Let the state $x(t)$ of the actual system in~\cref{eqn:proset:dynamics} with control input~\cref{eqn:control_def:overall_control} satisfy 
$x(t) \in \Omega(\rho, x^\star(t))$ for all $t \in [0, \tau]$, for some $\tau>0$. Additionally, let Assumptions~\ref{assmp:desired_control}-\ref{assmp:moore_penrose} hold. 
Then the state prediction error $\tilde{x}(t)$ defined in~\cref{eqn:control_def:state_predictor} satisfies
\begin{align}
    \label{eqn:app:tildebound1}
    \norm{\tilde{x}}_{\mathcal{L}_\infty}^{[0,\tau]} &\le \frac{\Delta_{\tilde{x}}}{\sqrt{\Gamma}}.
\end{align}
Define $\tilde{\sigma}(t) := \hat{\sigma}(t) - h(t,x)$ and $\tilde{\eta}(s) := C(s)\tilde{\sigma}(s)$. Then the following inequality holds
\begin{align}
    \label{eqn:app:tildebound2}
    \norm{\tilde{\eta}}_{\mathcal{L}_\infty}^{[0,\tau]} &\le \frac{\Delta_{\tilde{\eta}}}{\sqrt{\Gamma}},
\end{align}
where $\Delta_{\tilde{x}}$ and $\Delta_{\tilde{\eta}}$ are defined in \cref{eqn:bounds:tx,eqn:bounds:teta} respectively.
\end{lemma}
\begin{proof}
The state predictor error dynamics are computed using \cref{eqn:proset:dynamics} and \cref{eqn:control_def:state_predictor} as
\[
\dot{\tilde{x}}(t) = \dot{\hat{x}}(t) - \dot{x}(t) = A_m \tilde{x}(t) + B(x)\tilde{\sigma}(t), \quad \tilde{x}(0) = 0,
\]
where $\tilde{\sigma}(t) = \hat{\sigma}(t) - h(t,x)$. Consider the Lyapunov function $V(\tilde{x}, \tilde{\sigma}) = \tilde{x}(t)^\top P \tilde{x}(t) + \tilde{\sigma}(t)^\top \Gamma^{-1} \tilde{\sigma}(t)$. Then its time derivative is given by
\[
\dot{V}(\tilde{x}, \tilde{\sigma}) = -\tilde{x}(t)^\top Q \tilde{x}(t) + 2\tilde{\sigma}(t)^\top B(x)^\top P \tilde{x}(t) +  2\tilde{\sigma}(t)^\top \Gamma^{-1} (\dot{\hat{\sigma}}(t) - \dot{h}(t, x)),
\]
where $P \succ 0$ and $Q \succ 0$ define the adaptation law in~\cref{eqn:control_def:adaptation_law}. Combing the adaptation law with the equation above, we obtain
\[
\dot{V}(\tilde{x}, \tilde{\sigma}) = -\tilde{x}(t)^\top Q \tilde{x}(t) + 2\tilde{\sigma}(t)^\top (B(x)^\top P \tilde{x}(t) + \Proj_{\mathcal{H}}(\hat{\sigma}(t), -B(x)^\top P \tilde{x}(t)) -  2\tilde{\sigma}(t)^\top \Gamma^{-1} \dot{h}(t, x).
\]
From \cite[Lemma 6]{lavretsky2011projection}, the projection operator ensures that $\tilde{\sigma}(t)\left(\Proj_{\mathcal{H}}(\hat{\sigma}(t), y) - y\right) \le 0$ for all $y \in \mathbb{R}^n$. Therefore, the equation above reduces to
\[
\dot{V}(\tilde{x}, \tilde{\sigma}) \le -\tilde{x}(t)^\top Q \tilde{x}(t) - 2\tilde{\sigma}(t)^\top \Gamma^{-1}\dot{h}(t, x).
\]
It is easy to show that $\norm{\tilde{\sigma}(t)} \le 2\Delta_h$. Since $\norm{\dot{x}}_{\mathcal{L}_\infty}^{[0, \tau]} \le \Delta_{\dot{x}}$ from \cref{eqn:bounds:dx}, we have the following bound for the time-derivative of the uncertainty
\[
\norm{\dot{h}(t,x)} = \norm{\pdv{h}{t}(x) + \pdv{h}{x}(x) \dot{x}} \le \Delta_{h_t} + \Delta_{h_x}\Delta_{\dot{x}},
\]
for all $t \in [0, \tau]$. Combining, these bounds with the equation above results in the following inequality
\[
\dot{V}(\tilde{x}, \tilde{\sigma}) \le -\underline{\lambda}(Q)\norm{\tilde{x}}^2 + 4\Delta_h\Gamma^{-1}(\Delta_{h_t} + \Delta_{h_x}\Delta_{\dot{x}}).
\]
Now, if $\dot{V}(\tilde{x}(t), \tilde{\sigma}(t)) \ge 0$, then
\begin{equation}
\label{eqn:app:tildebound3}
\norm{\tilde{x}(t)}^2 \le \frac{4 \Delta_h(\Delta_{h_t} + \Delta_{h_x}\Delta_{\dot{x}})}{\Gamma\underline{\lambda}(Q)},
\end{equation}
for all $t \in [0, \tau]$. However, the Lyapunov function always remains bounded as
\[
\ulambda(P)\norm{\tilde{x}(t)}^2 \le V(\tilde{x}(t), \tilde{\sigma}(t)) \le \olambda(P)\norm{\tilde{x}(t)}^2 + 4\Delta_h^2 \Gamma^{-1}.
\]
Combining the equation above with \cref{eqn:app:tildebound3}, we arrive at the result in \cref{eqn:app:tildebound2}. In order to show that the inequality in \cref{eqn:app:tildebound2} holds, we start with the error in the uncertainty estimate from the state predictor error dynamics in terms of the Moore-Penrose inverse of $B(x)$ as
\[
\tilde{\sigma}(t) = B^\dagger(x)\dot{\tilde{x}}(t) - B^\dagger(x)A_m \tilde{x}(t).
\]
Since $\tilde{\eta}(s) = C(s)\tilde{\sigma}(s)$, we obtain
\begin{align*}
\tilde{\eta}(s) &= C(s)\mathscr{L}[B^\dagger(x)\dot{\tilde{x}}(t) - B^\dagger(x)A_m \tilde{x}(t)] \\
&= C(s)\mathscr{L}\left[\odv{}{t}\left(B^\dagger(x)\tilde{x}(t)\right) -\partial_{\dot{x}} B^\dagger(x)\tilde{x}(t) - B^\dagger(x)A_m \tilde{x}(t)\right] \\
&= C(s)s\mathscr{L}\left[B^\dagger(x)\tilde{x}(t)\right] + C(s)\mathscr{L}\left[-\partial_{\dot{x}} B^\dagger(x)\tilde{x}(t) - B^\dagger(x)A_m \tilde{x}(t)\right].
\end{align*}
Since $\norm{\dot{x}}_{\mathcal{L}_\infty}^{[0, \tau]} \le \Delta_{\dot{x}}$ from \cref{eqn:bounds:dx},  we have that $\norm{\partial_{\dot{x}}B^\dagger(x)} = \norm{\sum_{j=0}^n\pdv{B^\dagger(x)}{x_j}\dot{x}_j(t)} \le \Delta_{B^\dagger_x}\Delta_{\dot{x}}$ for all $t \in [0, \tau]$. Furthermore, from \cref{assmp:moore_penrose} we have that $\norm{B^\dagger(x)} \le \Delta_{B^\dagger}$ for all $x \in \Omega(\rho, x^\star(t))$ and $t \in [0, \tau]$. Therefore, using the property from \cite[Lemma A.7.1]{hovakimyan2010L1}, the following inequality holds
\[
\norm{\tilde{\eta}(t)} \le \left(\norm{C(s)s}_{\mathcal{L}_1} \Delta_{B^\dagger} + \norm{C(s)}_{\mathcal{L}_1} \left( \Delta_{B^\dagger_x} \Delta_{\dot{x}} + \Delta_{B^\dagger} \norm{A_m}\right) \right)\frac{\Delta_{\tilde{x}}}{\sqrt{\Gamma}},
\]
for all $x \in \Omega(\rho, x^\star(t))$ and $t \in [0, \tau]$. Since $\norm{C(s)}_{\mathcal{L}_1} = 1$, we arrive at the result.
\end{proof}
\begin{lemma}
\label{lem:dgeobound}
Let the state $x_r(t)$ of the reference system in~\cref{eqn:reference_system} and the state $x(t)$ of the real system in~\cref{eqn:proset:dynamics} with control input~\cref{eqn:control_def:overall_control} satisfy 
$x_r(t), x(t) \in \Omega(\rho, x^\star(t))$ for all $t \in [0, \tau]$, for some $\tau>0$. Additionally, let Assumptions~\ref{assmp:desired_control}-\ref{assmp:bounds_uncertainty} hold. Then if $\ogamma(\cdot,t): [0,1] \to \mathcal{X}$ is the minimizing geodesic under the metric $M(x)$ satisfying \cref{assmp:CCM} such that $\ogamma(0,t) = x(t)$ and $\ogamma(1,t) = x_r(t)$,
the following inequality is satisfied
\begin{equation}
    \label{eqn:app:dgeo0}
\norm{\odv{}{t} \left(B(x)^\top M(x) \ogamma_s(0,t)\right)} \le \Delta_{\dot{\Psi}} \norm{x_r(t) - x(t)},
\end{equation}
for all $t \in [0, \tau]$, where $\Delta_{\dot{\Psi}}$ is defined in \cref{eqn:bounds:dpsi}.
\end{lemma}
\begin{proof}
We apply chain rule and triangle inequality to obtain
\begin{align}
    \norm{\odv{}{t} \left(B(x)^\top M(x) \ogamma_s(0,t)\right)} \le \norm{B(x)^\top M(x) \dot{\overline{\gamma}}_s(0, t)} + \norm{B(x)^\top \partial_{\dot{x}(t)} M(x) \overline{\gamma}_s(0,t)}& \nonumber \\ + \norm{\partial_{\dot{x}(t)}B(x)^\top M(x) \overline{\gamma}_s(0,t)}&. \label{eqn:app:dgeo1} 
\end{align}
The time-evolution of the velocity of the minimizing geodesic evaluated at $s=0$ is given by the differential dynamics \cite[Theorem 3.2]{singh2019robust} of the actual system as follows
\begin{equation}
    \label{eqn:app:dgeo2}
    \dot{\ogamma}_s(0,t) = \left(\left[ \pdv{f(x)}{x} + \sum_{j=1}^m \left(u_j(t) + h_j(t, x)  \right)\pdv{b_j(x)}{x} \right] + B(x) \pdv{h(t, x)}{x} \right) \overline{\gamma}_s(0, t) + B(x)\delta_u,
\end{equation}
where $b_j(x)$ is the $j^{\textrm{th}}$ column of $B(x)$, $u_j(t)$ is the $j^{\textrm{th}}$ value in control channel, and $h_j(t, x)$ is the $j^{\textrm{th}}$ value of the uncertainty. Previously, from \cref{assmp:bounds_known,assmp:bounds_uncertainty} we know that
\begin{gather*}
    \norm{\pdv{f(x)}{x}} \le \Delta_{f_x}, \quad \norm{B(x)} \le \Delta_B, \quad \sum_{i=1}^{n} \norm{\pdv{B(x)}{x_i}} \le \Delta_{B_x}, \quad \sum_{j=1}^{m} \norm{\pdv{b_j(x)}{x}} \le \Delta_{b_x}, \\
    \norm{h(t, x)} \le \Delta_h, \quad \norm{\pdv{h(t,x)}{x}} \leq \Delta_{h_x}, 
\end{gather*}
for all $x \in \Omega(\rho, x^\star(t))$ and $t \in [0, \tau]$. Additionally from \cref{eqn:app:geobound2}, \cref{eqn:B7:controlbound},  and \cite[Theorem 5.2]{singh2019robust} the following hold
\[
    \norm{u(t)} \le \Delta_{u^\star} + \rho \Delta_{\delta_u} + \Delta_h, \quad \norm{\overline{\gamma}_s(1, t)} \le \sqrt{\frac{\overline{\alpha}}{\underline{\alpha}}} \|x_r(t) - x(t)\|, \quad \norm{\delta_u} \le \Delta_{\delta_u} \|x_r(t) - x(t)\|
\]
respectively for all $x \in \Omega(\rho, x^\star(t))$ and $t \in [0, \tau]$, where $\Delta_{\delta_u}$ defined in \cref{eqn:bounds:delta_u}. With these considerations, the expression in \cref{eqn:app:dgeo2} is bounded by $\Delta_{\dot{\ogamma}_s}\norm{x_r(t) - x(t)}$ for all $t \in [0, \tau]$. Furthermore, from \cref{assmp:CCM} we know that $\oalpha\mathbb{I}_n \succ M(x) \succ \ualpha\mathbb{I}_n \succ 0$ for all $x \in \mathcal{X}$, which implies that 
\begin{equation}
\label{eqn:app:dgeo3}
\norm{B(x)^\top M(x) \dot{\ogamma}_s(0,t)} \le \Delta_{B} \oalpha \Delta_{\dot{\ogamma}_s} \norm{x_r(t) - x(t)}, \quad \forall t \in [0, \tau].
\end{equation}
Since, $\norm{\partial_{\dot{x}}M(x)} \le \norm{\sum_{i=1}^n \pdv{M}{x_i} \dot{x}_i(t)} \le \Delta_{M_x} \Delta_{\dot{x}}$ from \cref{eqn:bounds:M_x} and \cref{eqn:bounds:dx} for all $x \in \Omega(\rho, x^\star(t))$ and $t \in [0, \tau]$, we obtain
\begin{equation}
\label{eqn:app:dgeo4}
\norm{B(x)^\top \partial_{\dot{x}} M(x) \ogamma_s(0,t)} \le \Delta_{B} \Delta_{M_x} \Delta_{\dot{x}} \sqrt{\frac{\oalpha}{\ualpha}} \norm{x_r(t) - x(t)}, \quad \forall t \in [0, \tau].
\end{equation}
Finally, $\norm{\partial_{\dot{x}}B(x)} \le \norm{\sum_{i=1}^n \pdv{B}{x_i} \dot{x}_i(t)} \le \Delta_{B_x} \Delta_{\dot{x}}$ from \cref{assmp:bounds_known} and \cref{eqn:bounds:dx}, and $\norm{M(x)\ogamma_s(0,t)} \le \oalpha \norm{x_r(t) - x(t)}$ holds from \cref{eqn:app:geobound1}. Therefore
\begin{equation}
\label{eqn:app:dgeo5}
\norm{\partial_{\dot{x}} B(x)^\top M(x) \ogamma_s(0,t)} \le \Delta_{B_x} \Delta_{\dot{x}} \oalpha \norm{x_r(t) - x(t)}, \quad \forall t \in [0, \tau].
\end{equation}
Substituting \cref{eqn:app:dgeo3,eqn:app:dgeo4,eqn:app:dgeo5} in \cref{eqn:app:dgeo1}, we arrive at the main result in \cref{eqn:app:dgeo0}.
\end{proof}

\section{Main Results}\label{app:main}

In this appendix we provide the proofs of the claims in the manuscript.

\begin{proof}[Proof of~\cref{lem:refbound}]
We show that $\norm{x^\star - x_r}_{\mathcal{L}_\infty} < \rho_r$ by contradiction. We have $\norm{x^\star(0) - x_r(0)} = \norm{x^\star_0 - x_0} < \rho_r$, since $\epsilon > 0$ and $\frac{\overline{\alpha}}{\underline{\alpha}} \ge 1$ from \cref{eqn:conditions:rho_r}. Assume that the lemma statement does not hold; then since $x_r$ and $x^\star$ are continuous, there must exist a $\tau^\star > 0$ such that
\begin{align*}
    \norm{x^\star(\tau^\star)  - x_r(\tau^\star)} &= \rho_r, \\
    \norm{x^\star(t) - x_r(t)} &< \rho_r, \quad t \in [0, \tau^\star).
\end{align*}
Consider the minimizing geodesic $\overline{\gamma}(s, t)$ such that $\overline{\gamma}(0, t) = x^\star(t)$ and $\overline{\gamma}(1, t) = x_r(t)$. From \cite[Theorem 3.2]{singh2019robust}, the time derivative of the Riemannian energy satisfies
\begin{align*}
    \frac{1}{2}\dot{\mathcal{E}}(x^\star, x_r)
    &= \overline{\gamma}_s(s, t)^\top M(\overline{\gamma}(s, t)) \dot{\overline{\gamma}}(s, t) \vert_{s=0}^{s=1} \\
    &= \overline{\gamma}_s(1, t)^\top M(x_r) \dot{x}_r(t) - \overline{\gamma}_s(0, t)^\top M(x^\star) \dot{x}^\star(t).
\end{align*}
Substituting the desired and reference system dynamics from \cref{eqn:proset:unperturbed_dynamics,eqn:reference_system} respectively into the equation above, we obtain
\begin{align}
\begin{split}
    \frac{1}{2}\dot{\mathcal{E}}(x^\star, x_r) =  \overline{\gamma}_s(1,t)^\top M(x_r) \left( f(x_r) + B(x_r)(u_{c,r}(t) + h(t, x_r) -\eta_r(t))\right)& \\ -\overline{\gamma}_s(0,t)^\top M(x^\star) \left( f(x^\star) + B(x^\star)u^\star(t) \right)&.
\end{split}
\label{eqn:app:ref1}
\end{align}
As shown in \cite[Lemma 1]{lopez2019contraction}, since the uncertainty is matched with the control channel, the metric $M$ designed for the unperturbed system also satisfies the stronger CCM conditions for the real system. This implies that the unperturbed part of the reference system is exponentially convergent with rate $\lambda$, given by the following expression
\begin{align*}
    \overline{\gamma}_s(1,t)^\top M(x_r)\left( f(x_r) + B(x_r)u_{c,r}(t) \right)
    -\overline{\gamma}_s(0,t)^\top M(x^\star)\left( f(x^\star) + B(x^\star)u^\star(t) \right)
    \le -\lambda\mathcal{E}(x^\star, x_r).
\end{align*}
Substituting into \cref{eqn:app:ref1} produces
\begin{align*}
        \frac{1}{2}\dot{\mathcal{E}}(x^\star, x_r) \le -\lambda\mathcal{E}(x^\star, x_r)
        + \overline{\gamma}_s(1,t)^\top M(x_r) B(x_r)(h(t, x_r)-\eta_r(t)) .
\end{align*}
Integrating both sides of the equation above, we have
\begin{align}
    \label{eqn:app:energy_integral}
    \mathcal{E}(x^\star, x_r) \le e^{-2\lambda t}\mathcal{E}(x^\star_0, x_0)
    + 2\int_0^t e^{-2\lambda(t-\nu)}\overline{\gamma}_s(1,t)^\top M(x_r)B(x_r)(h(\nu, x_r)-\eta_r(t))  \diff \nu.
\end{align}
The integral term on the right-hand side in the equation above can be expressed as the solution to the following virtual scalar system
\begin{align}
    \label{eqn:app:ref2a}
    \dot{z}(t) &= -2\lambda z(t) + \overline{\gamma}_s(1,t)^\top M(x_r) B(x_r) \xi(t), \quad z(0) = 0, \\
    \label{eqn:app:ref2b}
    \xi(s) &= (1 - C(s)) \mathscr{L}[h(t, x_r)].
\end{align}
Note that from \cref{lem:geobound,lem:dynbound,assmp:bounds_known,assmp:bounds_uncertainty} the following bounds hold for all $x_r(t) \in \Omega(\rho_r, x^\star(t)) \subset \Omega(\rho, x^\star)$ (since $\rho_r < \rho$) and all $t \in [0, \tau^\star]$:
\begin{gather*}
    \norm{\overline{\gamma}_s(1,t) M(x_r) B(x_r)} \le \rho \oalpha \Delta_B, \quad \norm{h(t,x_r)} \le \Delta_h, \\
    \norm{\dot{h}(t,x_r)} = \norm{\pdv{h(t,x_r)}{t} + \pdv{h(t,x_r)}{x} \dot{x}_r(t)} \le \Delta_{h_t} + \Delta_{h_x}\Delta_{\dot{x}_r}.
\end{gather*}
Then the solution of a linear system of the form in \cref{eqn:app:ref2a,eqn:app:ref2b} satisfies the following norm bound from \cref{lem:scalar}
\begin{align*}
    \norm{z(t)} \le \rho \oalpha \Delta_B \left( \frac{\Delta_h}{\abs{2\lambda - \omega}} +  \frac{\Delta_{h_t} + \Delta_{h_x}\Delta_{\dot{x}_r}}{2\lambda \omega} \right) = \frac{1}{2}\ualpha \zeta_1(\omega),
\end{align*}
where $\zeta_1$ is defined in \cref{eqn:conditions:zeta_1}. Substituting this inequality into \cref{eqn:app:energy_integral} produces
\[
    \mathcal{E}(x^\star, x_r) \le e^{-2\lambda t}\mathcal{E}(x^\star_0, x_0) + \ualpha \zeta_1(\omega)
\]
Moreover, since the metric $M$ satisfies \cref{assmp:CCM}, the Riemannian energy satisfies the following lower bound from \cref{lem:riemannbounds} 
\begin{align*}
\underline{\alpha} \norm{x^\star(t) - x_r(t)}^2 \le \mathcal{E}(x^\star, x_r)
\le e^{-2\lambda t}\mathcal{E}(x^\star_0, x_0) +  \ualpha \zeta_1(\omega).
\end{align*}
In our contradiction argument, we had that $\norm{x^\star(\tau^\star) - x(\tau^\star)} = \rho_r$. Therefore the following inequality must be satisfied
\begin{align}
\label{eqn:app:ref3a}
\ualpha \rho_r^2 &\le e^{-2\lambda \tau^\star}\mathcal{E}(x^\star_0, x_0) + \ualpha \zeta_1(\omega) \\
\label{eqn:app:ref3b}
&<  \mathcal{E}(x^\star_0, x_0) + \ualpha \zeta_1(\omega),
\end{align}
for some $\tau^\star > 0$. However, from~\cref{eqn:filter_condition_1} the bandwidth is chosen such that
\[
\ualpha \rho_r^2 \ge \mathcal{E}(x^\star_0, x_0) + \ualpha \zeta_1(\omega).
\]
This directly contradicts our earlier statement in \cref{eqn:app:ref3b}. Therefore, $\|x_r - x^\star\|_{\mathcal{L}_\infty} < \rho_r$. Moreover, from \cref{eqn:app:ref3a} we also obtain the following uniform ultimate bound
\[
\norm{x^\star(t) - x_r(t)} \le \sqrt{\frac{e^{-2\lambda T}\mathcal{E}(x^\star_0, x_0)}{\underline{\alpha}} + \zeta_1(\omega)},
\]
for all $t \ge T \ge 0$.
\end{proof}


\begin{proof}[Proof of Lemma~\ref{lem:realbound}]
We prove this lemma by contradiction. Assume that
\begin{equation}
\label{eqn:app:contradiction}
\norm{x_r - x}_{\mathcal{L}_\infty}^{[0, \tau]} \ge \rho_a,
\end{equation}
for some $\tau > 0$, where $\rho_a$ is an arbitrary positive scalar used in the definition of $\rho$ in \cref{eqn:conditions:rho}. Since $x_r(0) = x(0)$, there exists a $\tau^\star \in (0, \tau]$ such that
\begin{align}
    \norm{x_r(\tau^\star) - x(\tau^\star)} &= \rho_a, \nonumber \\
    \label{eqn:app:taustar}
    \norm{x_r(t) - x(t)} &< \rho_a, \quad t \in [0, \tau^\star).
\end{align}
Let $\ogamma(s,t)$ be the minimizing geodesic between $x_r$ and $x$ such that $\ogamma(1,t) = x_r(t)$ and $\ogamma(0,t) = x(t)$. Consider the Riemannian energy $\mathcal{E}(x_r,x)$ between $x_r$ and $x$. Then the time derivative of the Riemannian energy as shown in \cite[Theorem 3.2]{singh2019robust} is given by
\begin{align*}
    \frac{1}{2}\dot{\mathcal{E}}(x_r,x)
    &= \overline{\gamma}_s(s, t)^\top M(\overline{\gamma}(s, t)) \dot{\overline{\gamma}}(s, t) \vert_{s=0}^{s=1} \\
    &= \ogamma_s(1)^\top M(x_r)\dot{x}_r - \ogamma_s(0) M(x) \dot{x}.
\end{align*} 
Substituting the reference system dynamics from \cref{eqn:reference_system} and the actual system dynamics  from \cref{eqn:proset:dynamics} with the control law from \cref{eqn:control_def:overall_control}, we get
\begin{align*}
    \frac{1}{2}\dot{\mathcal{E}}(x_r,x) =& \gamma_s(1)^\top M(x_r) \left[
        f(x_r) + B(x_r)(u_{c,r}(t) + h(t,x_r) - \eta_r(t))
    \right] \\
    & - \gamma_s(0)^\top M(x)\left[
    f(x) + B(x)(u_c(t) + h(t,x) - \hat{\eta}(t))
    \right],
\end{align*}
where $\hat{\eta}(t) = -u_a(t)$. Similar to the reasoning used in \cref{lem:refbound}, the metric $M$ designed for the ideal system also satisfies the stronger CCM conditions for the real system~\cite[Lemma~1]{lopez2019contraction}. This implies that the nominal parts of the system are contracting, given by
\begin{equation*}
    \frac{1}{2}\dot{\mathcal{E}}(x_r,x) \leq -\lambda \mathcal{E}(x_r,x) +  \Psi(x_r)^\top(h(t, x_r) - \eta_r(t)) - \Psi(x)^\top (h(t, x) - \hat{\eta}(t)),
\end{equation*}
where $\Psi(x_r) := B(x_r)^\top M(x_r) \overline{\gamma}_s(1, t)$ and $\Psi(x) := B(x)^\top M(x)\overline{\gamma}_s(0, t)$ are introduced for clarity. Define $\eta(s) = C(s)\mathscr{L}[h(t,x)]$; then by adding and subtracting $\Psi(x)^\top(h(t,x_r) - \eta_r(t) + \eta(t))$ on the right-hand side we obtain
\begin{align*}
\begin{split}
    \frac{1}{2}\dot{\mathcal{E}}(x_r,x) \leq -\lambda \mathcal{E}(x_r,x) + (\Psi(x_r) - \Psi(x))^\top(h(t, x_r)- \eta_r(t))
    + \Psi(x)^\top \left(h(t, x_r) - \eta_r(t) - h(t,x) + \eta(t)\right)& \\
    + \Psi(x)^\top\left(\hat{\eta}(t) - \eta(t)\right)&.
\end{split}
\end{align*}
Since $h(t, x_r) - \eta_r(t) = \mathscr{L}^{-1}[(1 - C(s))\mathscr{L}[h(t, x_r)]]$, $h(t, x) - \eta(t) = \mathscr{L}^{-1}[(1 - C(s))\mathscr{L}[h(t, x)]]$, with $\tilde{\eta}(t)=\hat{\eta}(t) - \eta(t)$, we rewrite the equation above as
\begin{equation}\label{eqn:app:real1}
    \frac{1}{2}\dot{\mathcal{E}}(x_r,x) \leq -\lambda \mathcal{E}(x_r, x) + \Phi_1(x_r, x) + \Phi_2(x_r, x) + \Phi_3(x_r, x),
\end{equation}
where
\begin{align*}
    \Phi_1(x_r, x) &:=  (\Psi(x_r) - \Psi(x))^\top \mathscr{L}^{-1}[(1 - C(s))\mathscr{L}[h(t, x_r)]], \\
    \Phi_2(x_r, x) &:=  \Psi(x)^\top \mathscr{L}^{-1}[(1 - C(s))\mathscr{L}[h(t, x_r) -h(t, x)]], \\
    \Phi_3(x_r, x) &:=  \Psi(x)^\top \tilde{\eta}(t).
\end{align*}
Solving the differential equation in \cref{eqn:app:real1}  we obtain
\begin{align*}
    \mathcal{E}(x_r, x) \le \mathcal{E}(x_r(0), x(0)) + 2\int_0^t e^{-2\lambda(t-\nu)} (\Phi_1(x_r, x) + \Phi_2(x_r, x) +\Phi_3(x_r, x)) \diff \nu.
\end{align*}
Since $x_r(0) = x(0) \implies \mathcal{E}(x_r(0), x(0)) = 0$,  the inequality above reduces to
\begin{align}
    \label{eqn:app:real2}
    \mathcal{E}(x_r, x) \le  2\int_0^t e^{-2\lambda(t-\nu)} (\Phi_1(x_r, x) + \Phi_2(x_r, x) +\Phi_3(x_r, x)) \diff \nu.
\end{align}
Notice that $\norm{\Psi(x_r) - \Psi(x)}$ satisfies the following bounds
\begin{align*}
    \norm{\Psi(x_r) - \Psi(x)} &\le \norm{B(x_r)^\top M(x_r)\ogamma_s(1,t) - B(x)M(x)\ogamma_s(0,t)}.
\end{align*}
Adding and subtracting $B(x)M(x_r)\ogamma_s(1,t)$, from the right hand side of the equation above we obtain
\begin{align}
    \norm{\Psi(x_r) - \Psi(x)} &\le \norm{(B(x_r) - B(x))^\top M(x_r)\ogamma_s(1,t) + B(x)^\top(M(x_r)\ogamma_s(1,t) - M(x)\ogamma_s(0,t))} \nonumber \\
    \label{eqn:app:bound1}
    &\le \norm{B(x_r) - B(x)}\norm{M(x_r)\overline{\gamma}_s(1, t)}
    + \norm{B(x)}\norm{M(x_r)\overline{\gamma}_s(1, t) - M(x)\overline{\gamma}_s(0, t)}.
\end{align}
Since $x_r(t) \in \Omega(\rho_r, x^\star(t))$ from \cref{lem:refbound} and $x(t) \in \Omega(\rho, x^\star(t))$ for all $t \in [0, \tau^\star]$ by assumption in \cref{eqn:app:taustar}, the following bounds hold for $t \in [0, \tau^\star]$ as a result of \cref{assmp:bounds_known}:
\[
\norm{B(x_r) - B(x)} \le \Delta_{B_x}\norm{x_r(t) - x(t)}, \quad \norm{B(x)} \le \Delta_B.
\]
And from \cref{lem:geobound,lem:roadblock} we obtain the following respectively 
\[
\norm{M(x_r)\overline{\gamma}_s(1, t)} \le \oalpha \norm{x_r(t) - x(t)}, \quad \norm{M(x_r)\ogamma_s(1,t) - M(x)\ogamma_s(0,t)} \le \frac{\oalpha}{2\ualpha} \Delta_{M_x} \norm{x_r(t) - x(t)}^2,
\]
where $\Delta_{M_x}$ is defined in \cref{eqn:bounds:M_x}. Substituting these bounds in \cref{eqn:app:bound1} produces
\begin{align}
\label{eqn:app:bound2}
    \norm{\Psi(x_r) - \Psi(x)} \le  \frac{1}{2} \oalpha \Delta_{\Psi_x} \norm{x_r(t) - x(t)}^2,
\end{align}
which holds for all $t \in [0, \tau^\star]$ and where $\Delta_{\Psi_x}$ is a scalar defined in \cref{eqn:bounds:psi_x}. Additionally, from \cref{assmp:bounds_uncertainty} we know that the following inequalities hold
\begin{align}
    \label{eqn:app:bound3}
    \norm{h(t, x_r)} \le \Delta_h, \quad \norm{h(t, x_r) - h(t, x)} \le \Delta_{h_x} \norm{x_r(t) - x(t)}.
\end{align}
Since $\norm{\pdv{h}{t}(t,x_r)} \le \Delta_{h_t}$ and $\norm{\pdv{h}{x}(t,x_r)} \le \Delta_{h_x}$ from \cref{assmp:bounds_uncertainty}, and $\norm{\dot{x}_r}_{\mathcal{L}_\infty} \le \Delta_{\dot{x}_r}$ from \cref{lem:dynbound}, the following inequality is satisfied
\begin{align}
    \norm{\dot{h}(t,x_r)} &= \norm{\pdv{h}{t}(t, x_r) + \pdv{h}{x}(t, x_r)\dot{x}_r} \nonumber \\
    \label{eqn:app:bound4}
    &\le \Delta_{h_t} + \Delta_{h_x}\Delta_{\dot{x}_r}.
\end{align}
Since $\norm{M(x)\ogamma(0,t)} \le \oalpha \norm{x_r(t) - x(t)}$ from \cref{eqn:app:geobound1} for all $t \in [0, \tau^\star]$, the following holds
\begin{align}
    \norm{\Psi(x)} &= \norm{B(x)^\top M(x) \ogamma(0,t)}
    \le \norm{B(x)} \norm{M(x)\ogamma(0,t)} \nonumber \\
    \label{eqn:app:bound5}
    &\le \Delta_B \oalpha \norm{x_r(t) - x(t)}
\end{align}
for all $t \in [0, \tau^\star]$. From \cref{lem:dgeobound} we have the following result for all $t \in [0, \tau^\star]$
\begin{equation}
    \label{eqn:app:bound6}
    \norm{\dot{\Psi}(x)} \le \Delta_{\dot{\Psi}}\norm{x_r(t) - x(t)}.
\end{equation}
In order to derive bounds on \cref{eqn:app:real2}, define the following scalar trajectories
\begin{align*}
z_1(t) := \int_0^t e^{-2\lambda(t - \nu)}\Phi_1(x,x_r)\diff\nu, \quad
z_2(t) := \int_0^t e^{-2\lambda(t - \nu)}\Phi_2(x,x_r)\diff\nu.
\end{align*}
Then, the functions $z_i$, $i \in \{1,2\}$, are the states of the following systems
\begin{subequations}\label{eqn:final_boss:scalar_systems}
\begin{align}
    \dot{z}_i(t) &=  -2\lambda z_i(t) + b_i(t)\xi_i(t), \quad z_i(0) = 0,\\
    \xi_i(s) &= (1-C(s))\sigma_i(s),
\end{align}
\end{subequations} where
\begin{align*}
    b_1(t) =& \Psi(x_r) - \Psi(x),\quad \sigma_1(t) = h(t,x_r),\\
    b_2(t) =& \Psi(x), \quad \sigma_2(t) = h(t,x_r) - h(t,x).
\end{align*}
From \cref{eqn:app:contradiction} we assumed that $\norm{x_r(t) - x(t)} \le \rho_a$ for $t \in [0, \tau^\star]$. Using \cref{lem:scalar} for the $z_1(t)$ system, \cref{lem:scalar2} for the $z_2(t)$ system, and the bounds in \cref{eqn:app:bound2,eqn:app:bound3,eqn:app:bound4,eqn:app:bound5,eqn:app:bound6}, we have the following inequalities
\begin{align}
\label{eqn:app:tbound1}
   \norm{z_1(t)}\le  \frac{1}{2} \zeta_2(\omega) \rho_a^{2}, \quad
   \norm{z_2(t)}\le  \frac{1}{2} \zeta_3(\omega) \rho_a^{2},
\end{align}
for all $t \in [0, \tau^\star]$, and where $\zeta_2$ and $\zeta_3$ are defined in \cref{eqn:conditions:zeta_2,eqn:conditions:zeta_3} respectively. Moreover, it is easy to show from \cref{eqn:app:tildebound2} that
\begin{align}
\label{eqn:app:tbound2}
\norm{\int_0^t e^{-2\lambda(t - \nu)}\Phi_3(x,x_r)d\nu} \le \frac{\Delta_\theta\rho_a}{2\sqrt{\Gamma}},
\end{align}
for all $t \in [0, \tau^\star]$, where $\Delta_{\theta}$ is defined in \cref{eqn:bounds:theta}. Substituting \cref{eqn:app:tbound1,eqn:app:tbound2} into \cref{eqn:app:real1} we obtain the following bound on the Riemannian energy
\[
\mathcal{E}(x_r, x) \le \zeta_2(\omega)\rho_a^2 + \zeta_3(\omega)\rho_a^2 + \frac{\Delta_\theta \rho_a}{\sqrt{\Gamma}},
\]
for all $t \in [0, \tau^\star]$. Recall from the contradiction argument that $\norm{x_r(\tau^\star) - x(\tau^\star)} = \rho_a$. Since $\mathcal{E}(x_r, x) \ge \ualpha \norm{x_r(t) - x(t)}^2$ from \cref{lem:riemannbounds}, the following inequality must be satisfied at $t = \tau^\star$
\begin{align*}
\underline{\alpha}\rho_a^2 \le \zeta_2(\omega)\rho_a^2 + \zeta_3(\omega)\rho_a^2 + \frac{\Delta_\theta \rho_a}{\sqrt{\Gamma}} \implies
\sqrt{\Gamma} \le
\frac{\Delta_\theta}{\rho_a(\underline{\alpha} - \zeta_2(\omega) - \zeta_3(\omega))}.
\end{align*}
However, from~\cref{eqn:adaptation_rate_condition} the adaptation rate is chosen such that
\[
\sqrt{\Gamma} >
\frac{\Delta_\theta}{\rho_a(\underline{\alpha} - \zeta_2(\omega) - \zeta_3(\omega))}.
\]
This directly contradicts our earlier statement, and therefore $\|x_r - x\|_{\mathcal{L}_\infty}^{[0, \tau]} < \rho_a$.
\end{proof}

\begin{proof}[Proof of Theorem~\ref{thm:main_theorem}]
We will prove this by contradiction: assume that
\[
\norm{x^\star - x}_{\mathcal{L}_\infty} \ge \rho.
\]
From the definition of $\rho_r$ in \cref{eqn:conditions:rho_r}, we have $\norm{x^\star(0) - x(0)} < \rho_r$, which implies that $\norm{x^\star(0) - x(0)} < \rho$. Then there must exist a $\tau^\star > 0$ such that
\begin{align}
    \label{eqn:app:main1}
    \norm{x^\star(\tau^\star) - x(\tau^\star)} &= \rho, \\
    \norm{x^\star(t) - x(t)} &< \rho, \quad t \in [0, \tau^\star).
\end{align}
According to \cref{lem:refbound}, we have $\norm{x^\star - x_r}_{\mathcal{L}_\infty} < \rho_r$, and from \cref{lem:realbound} it follows that $\norm{x_r - x}_{\mathcal{L}_\infty}^{[0, \tau^\star]} < \rho_a$. Therefore, from triangle inequality  we have $\norm{x^\star - x}_{\mathcal{L}_\infty}^{[0, \tau^\star]} < \rho_r + \rho_a = \rho$, which implies that $\norm{x^\star(\tau^\star) - x(\tau^\star)} < \rho$. This directly contradicts our assumption in \cref{eqn:app:main1}, and therefore $\norm{x^\star - x}_{\mathcal{L}_\infty} < \rho$.
Additionally, from \cref{lem:refbound} the reference system satisfies a uniform ultimate bound given by $\norm{x^\star(t) - x(t)} \le \epsilon(\omega, T)$ for all $t \ge T \ge 0$. Therefore, $\norm{x^\star(t) - x(t)} < \epsilon(\omega, T) + \rho_a$ for all $t \ge T \ge 0$.
\end{proof}

\bibliographystyle{IEEEtran}
\bibliography{ref}

\end{document}